\documentclass[12pt]{article}
\usepackage{amssymb}
\usepackage{amsmath}
\usepackage{amsthm}
\usepackage{graphicx}
\usepackage[linesnumbered,algoruled,boxed,lined]{algorithm2e}
\usepackage{complexity} 
\usepackage{enumitem}
\usepackage{mathtools}

\usepackage{tikz}
\usetikzlibrary{positioning,backgrounds,patterns,calc}

\tikzset{
  circ/.style = {circle,draw,fill,inner sep=1.3pt}
}  

\usepackage{cleveref}
\usepackage{subcaption}

\usepackage[a4paper, total={6.0in, 8.5in}]{geometry}

\usepackage{color}

\newcommand{\calC}{\mathcal{C}}

\newcommand{\calF}{\mathcal{F}}

\DeclareMathOperator{\dist}{dist}

\newcommand{\yes}{\textsc{Yes}}
\newcommand{\no}{\textsc{No}}

\newtheorem{theorem}{Theorem}
\newtheorem{lemma}{Lemma}

\newtheorem{claim}{Claim}[section]
\newtheorem{observation}{Observation}[section]

\newtheorem{reduction rule}{Reduction Rule}[section]
\newtheorem{marking-scheme}{Marking Scheme}[section]

\usepackage{authblk}
\usepackage[numbers,sort]{natbib}

\newcommand{\mdfull}{\textsc{Metric Dimension}}

\newcommand{\sign}{\texttt{sign}}

\newcommand{\OO}{\mathcal{O}}

\newenvironment{claimproof}{\begin{proof}\renewcommand{\qedsymbol}{\claimqed}}{\end{proof}\renewcommand{\qedsymbol}{\plainqed}}
\let\plainqed\qedsymbol

\begin{document}

\title{Metric Dimension Parameterized by Feedback Vertex Set and Other Structural Parameters\thanks{An extended abstract of this paper was presented in~\cite{MFCSversion}.}~\thanks{Research supported by the European Research Council (ERC) consolidator grant No.~725978 SYSTEMATICGRAPH.}}
\date{}
\author[1]{Esther Galby}
\author[2]{Liana Khazaliya}
\author[1]{Fionn {Mc Inerney}}
\author[3]{Roohani Sharma}
\author[4]{Prafullkumar Tale}
\affil[1]{CISPA Helmholtz Center for Information Security, Saarbr\"ucken, Germany}
\affil[2]{Saint Petersburg State University, Saint Petersburg, Russia}
\affil[3]{Max Planck Institute for Informatics, Saarland Informatics Campus, Saarbr\"{u}cken, Germany}
\affil[4]{Indian Institute of Science Education and Research Pune, Pune, India}

\maketitle

\begin{abstract}
For a graph $G$, a subset $S \subseteq V(G)$ is called a \emph{resolving set} if for any two vertices $u,v \in V(G)$, there exists a vertex $w \in S$ such that $d(w,u) \neq d(w,v)$.
The {\sc Metric Dimension} problem takes as input a graph $G$ and a positive integer $k$, and asks whether there exists a resolving set of size at most $k$.
This problem was introduced in the 1970s and is known to be \NP-hard~[GT~61 in Garey and Johnson's book].
In the realm of parameterized complexity, Hartung and Nichterlein~[CCC~2013] proved that the problem is \W[2]-hard when parameterized by the natural parameter $k$.
They also observed that it is \FPT\ when parameterized by the vertex cover number and asked about its complexity under \emph{smaller} parameters, in particular the feedback vertex set number.
We answer this question by proving that {\sc Metric Dimension} is \W[1]-hard when parameterized by the combined parameter feedback vertex set number plus pathwidth.
This also improves the result of Bonnet and Purohit~[IPEC 2019] which states that the problem is \W[1]-hard parameterized by the pathwidth.
On the positive side, we show that {\sc Metric Dimension} is \FPT\ when parameterized by either the distance to cluster or the distance to co-cluster, both of which are smaller parameters than the vertex cover number.
\end{abstract}

\section{Introduction}

Problems dealing with distinguishing the vertices of a graph have attracted a lot of attention over the years, with the metric dimension problem being a classic one that has been vastly studied since its introduction in the 1970s by Slater~\cite{Slater75}, and independently by Harary and Melter~\cite{HM76}.    
Formally, given a graph $G$ and an integer $k\geq 1$, the {\sc Metric Dimension} problem asks whether there exists a subset $S\subseteq V(G)$ of vertices of $G$ of size at most $k$ such that, for any two vertices $u,v\in V(G)$, there exists a vertex $w\in S$ such that $d(w,u)\neq d(w,v)$.
If such a subset $S\subseteq V(G)$ exists, it is called a {\it resolving set}.
The size of a smallest resolving set of a graph $G$ is the {\it metric dimension} of $G$, and is denoted by $MD(G)$.

There are many variants and problems associated to the metric dimension, with {\it identifying codes}~\cite{KarpovskyCL98}, {\it adaptive identifying codes}~\cite{BGLM08}, and {\it locating dominating sets}~\cite{Slater87} asking for the vertices to be distinguished by their neighborhoods in the subset chosen. 
Other variants of note are the {\it $k$-metric dimension}, where each pair of vertices must be resolved by $k$ vertices in $S\subseteq V(G)$ instead of just one~\cite{ERY15}, and the {\it truncated metric dimension}, where the distance metric is the minimum of the distance in the graph and some integer $k$~\cite{TFL21}. 
Along similar lines, in the {\it centroidal dimension} problem, each vertex must be distinguished by its relative distances to the vertices in $S\subseteq V(G)$~\cite{FoucaudKS14}.
The metric dimension has also been considered in digraphs, with Bensmail et al.~\cite{BMNjournal} providing a summary of the related work in this area.
Interestingly, there are many game-theoretic variants of the metric dimension, such as {\it sequential metric dimension}~\cite{BMM+19}, the {\it localization game}~\cite{BGG18b,HaslegraveJK18}, and the {\it centroidal localization game}~\cite{BGG18a}. The metric dimension and its variants have been studied for both their theoretical interest and their numerous applications such as in network verification~\cite{BEE+06}, fault-detection in networks~\cite{UTS04}, pattern recognition and image processing~\cite{MT84}, graph isomorphism testing~\cite{B80}, chemistry~\cite{CEJ00,J93}, and genomics~\cite{TL19}. For more on these variants and others, see~\cite{KY21} for the latest survey.

Much of the related work around the metric dimension problem focuses on its computational complexity.
{\sc Metric Dimension} was first shown to be \NP-complete in general graphs in~\cite{GJ79}. 
Later, it was also shown to be \NP-complete in split graphs, bipartite graphs, co-bipartite graphs, and line graphs of bipartite graphs in~\cite{ELW15}, in bounded-degree planar graphs~\cite{DiazPSL17}, and interval and permutation graphs of diameter~$2$~\cite{FoucaudMNPV17b}. 
On the positive side, there are linear-time algorithms for {\sc Metric Dimension} in trees~\cite{Slater75}, cographs~\cite{ELW15}, cactus block graphs~\cite{HEW16}, and bipartite distance-hereditary graphs~\cite{Moscarini22}, and a polynomial-time algorithm for outerplanar graphs~\cite{DiazPSL17}. 

Since the problem is \NP-hard even for very restricted cases, it is natural to ask for ways to confront this hardness.
In this direction, the parameterized complexity paradigm allows for a more refined analysis of the problem’s complexity.
In this setting, we associate each instance $I$ with a parameter $\ell$, and are interested in an algorithm with running time $f(\ell) \cdot |I|^{\OO(1)}$ for some computable function $f$.
Parameterized problems that admit such an algorithm are called fixed parameter tractable (\FPT) with respect to the parameter under consideration.
On the other hand, under standard complexity assumptions, parameterized problems that are hard for the complexity class \W[1] or \W[2] do not admit such fixed-parameter algorithms.
A parameter may originate from the formulation of the problem itself (called \emph{natural parameters}) or it can be a property of the input graph (called \emph{structural parameters}).

Hartung and Nichterlein~\cite{HartungN13} proved that {\sc Metric Dimension} is \W[2]-hard when parameterized by the natural parameter, the solution size $k$, even when the input graph is bipartite and has maximum degree $3$.
This motivated the study of the parameterized complexity of the problem under structural parameterizations.
It was observed in~\cite{HartungN13} that the problem admits a simple \FPT\ algorithm when parameterized by the vertex cover number.
Gutin et al.~\cite{DBLP:journals/tcs/GutinRRW20} proved that the problem does not admit a polynomial kernel when parameterized by the vertex cover number unless $\NP \subseteq \coNP/poly$.
It took a considerable amount of work and/or meta-results to prove that there are \FPT~algorithms parameterized by the max leaf number~\cite{E15}, the modular width and the treelength plus the maximum degree~\cite{BelmonteFGR17}, and the treedepth and the combined parameter clique-width plus diameter~\cite{GHK22}. 
In~\cite{ELW15}, they gave an \XP~algorithm parameterized by the feedback edge set number.  
Only recently, it was shown that {\sc Metric Dimension} is \W[1]-hard parameterized by the pathwidth (and so, treewidth)~\cite{BP21}, answering an open question mentioned in~\cite{BelmonteFGR17, DiazPSL17,E15}. 
This result was improved upon since, with it being shown that {\sc Metric Dimension} is even \NP-hard in graphs of pathwidth~$24$~\cite{LM22}. 
For more on the metric dimension, see~\cite{TFL21survey} for a recent survey.

\begin{figure}[t]%
\footnotesize{
\definecolor{green}{RGB}{97, 191, 135}
\definecolor{red}{RGB}{226, 93, 66}
\definecolor{yellow}{RGB}{243, 208, 53}
\newcommand{\circgreen}{\raisebox{0.5pt}{\tikz{\node[draw,scale=0.4,circle,fill=yellow](){};}}}
\newcommand{\circred}{\raisebox{0.5pt}{\tikz{\node[draw,scale=0.4,circle,fill=yellow](){};}}}
\newcommand{\circyellow}{\raisebox{0.5pt}{\tikz{\node[draw,scale=0.4,circle,fill=yellow](){};}}}
\tikzstyle{open}=[align=center, rectangle, minimum height=.8cm,text width=2cm,fill=white!60,rounded corners=2mm,draw]
\tikzstyle{para}=[align=center, rectangle, minimum height=.8cm,text width=2cm,fill=red,rounded corners=2mm, draw]
\tikzstyle{w1}=[align=center, rectangle, minimum height=.8cm,text width=2cm,fill=yellow,rounded corners=2mm,draw]
\tikzstyle{fpt}=[align=center, rectangle, minimum height=.8cm,text width=2cm,fill=green,rounded corners=2mm,draw]
\tikzstyle{xp}=[align=center, rectangle, minimum height=.8cm,text width=2cm,fill=cyan,rounded corners=2mm,draw]
\resizebox{\textwidth}{!}{
\begin{tikzpicture}[node distance=7mm]

	        \node[fpt] (vc) at (7.5, 6) {Vertex \\ Cover \\ \cite{HartungN13}};
		\node[fpt, text width=1.75cm] (ml) at (12, 6)  {Max Leaf Number \\ \cite{E15}};
		\node[fpt, text width=1.6cm] (dc) at (2.5, 6)  {Distance \\ to Clique};

		\node[para] (mcc) at (0, 4) {Minimum Clique Cover \\ \cite{ELW15}};
 		\node[fpt, line width = 2pt,text width=2.1cm] (dcc) at (2.5, 4) {Distance \\ to Co-Cluster};
 		\node[fpt, line width = 2pt, text width=1.7cm] (dcl) at (4.9, 4) {Distance\\ to Cluster};
 		\node[open,text width=2.3cm] (ddp) at (7.5, 4) {Distance to Disjoint Paths};
 		\node[xp,text width=1.5cm] (fes) at (10, 4) {Feedback Edge Set \\ \cite{ELW15}};
	 	\node[fpt,text width=1.8cm] (td) at (12.1, 4) {Treedepth \\ \cite{GHK22}};
 		\node[open,text width=1.85cm] (bw) at (14.3, 4) {Bandwidth};

 		\node[para] (is) at (0, 1.8) {Maximum \\ Independent Set \\ \cite{ELW15}};
 		\node[open,text width=1.8cm] (dcg) at (2.5, 2) {Distance to Cograph};
 		\node[para,text width=1.9cm] (dig) at (4.9, 2) {Distance to Interval \\ \cite{FoucaudMNPV17b}};
 		\node[w1, line width = 2pt, text width=1.7cm] (fvs) at (7.5, 2) {Feedback Vertex Set};
 		\node[para, line width = 2pt, text width=1.6cm] (pw) at (12.1, 2) {Pathwidth \\ \cite{BP21,LM22}};
 		\node[para,text width=1.7cm] (mxd) at (14.3, 2) {Maximum Degree \\ \cite{DiazPSL17}};
 		
 		\node[para,text width=1.8cm] (dpf) at (2.5, 0) {Distance to Perfect \\ \cite{ELW15}};
 		\node[para,text width=1.6cm] (tw) at (12.1, 0) {Treewidth \\ \cite{BP21,LM22}};
 		
 		\node[circ] (ker_vc) at (8.45, 6.55) {};
 		\draw (8.35, 6.45) -- (8.55, 6.65);
 		\node[circ] (ker_dc) at (3.25, 6.35) {};
 		\draw (3.15, 6.25) -- (3.35, 6.45);
 		\node[circ] (ker_vc) at (3.47, 4.26) {};
 		\draw (3.37, 4.16) -- (3.57, 4.36);
 		\node[circ] (ker_vc) at (5.69, 4.26) {};
 		\draw (5.59, 4.16) -- (5.79, 4.36);
 		\node[circ] (ker_td) at (12.96, 4.35) {};
 		\draw (12.86, 4.25) -- (13.06, 4.45);
 		\node[circ] (ker_bw) at (15.15, 4.25) {};		
 		\draw (15.05, 4.15) -- (15.25, 4.35);				
 		\node[circ, fill=white, inner sep=1.5pt] (ker_ml) at (12.85, 6.55) {};
 		
 	    \draw (dc) .. controls (1.5, 5) and (0, 5) .. (mcc)
 	    (dc) -- (dcc)
 	    (dc) .. controls (3.5, 5) and (4.5, 5) .. (dcl)
 	    (vc) .. controls (6.3, 5) and (3, 5) .. (dcc)
 	    (vc) .. controls (6.8, 5) and (5.5, 5) .. (dcl)
 	    (vc) -- (ddp)
 	    (vc) .. controls (8.5, 5) and (12, 5) .. (td)
 	    (ml) .. controls (11, 5) and (8, 5) .. (ddp)
 	    (ml) .. controls (11.5, 5) and (10, 5) .. (fes)
 	    (ml) .. controls (12.5, 5) and (14.25, 5) .. (bw)
 	    (mcc) -- (is)
 	    (dcc) -- (dcg)
 	    (dcl) -- (dig)
 	    (ddp) -- (fvs)
 	    (td) -- (pw)
 	    (pw) -- (tw)
 	    (bw) -- (mxd)
 	    (ddp) .. controls (8, 3) and (11.5, 3) .. (pw)
 	    (fvs) .. controls (8, 1) and (11.5, 1.25) .. (tw)
 	    (bw) .. controls (13.75, 3) and (12.75, 3) .. (pw)
 	    (fes) .. controls (10, 3) and (8, 3) .. (fvs)
 	    (ddp) .. controls (7, 3) and (5.5, 3) .. (dig)
 	    (dcl) .. controls (4.5, 3) and (3, 3) .. (dcg)
 	    (dcg) -- (dpf)
 	    (dig) .. controls (4.75, 1) and (2.8, 1) .. (dpf)
		(fvs) .. controls (7, 1) and (3.6, 1) .. (dpf);
 	    
 	    \node at (2.25, 7.1) (colors) {
 	    \fcolorbox{green}{green}{\rule{0pt}{3pt}\rule{3pt}{0pt}} \ --- \FPT;
 	    \fcolorbox{cyan}{cyan}{\rule{0pt}{3pt}\rule{3pt}{0pt}} \ --- \XP;
 	    \fcolorbox{yellow}{yellow}{\rule{0pt}{3pt}\rule{3pt}{0pt}} \ --- \W[1];
 	    \fcolorbox{red}{red}{\rule{0pt}{3pt}\rule{3pt}{0pt}} \ --- para-\NP.};
 		
	\end{tikzpicture}
}%
}
\caption{Hasse diagram of graph parameters and associated results for {\sc Metric Dimension}. 
An edge from a lower parameter to a higher parameter indicates that the lower one is upper bounded by a function of the higher one. 
Colors correspond to the known hardness with respect to the highlighted  parameter. 
The parameters for which the hardness remains an open question are not colored. 
The crossed bold circle in the upper-right corner means that \textsc{Metric Dimension} does not admit a polynomial kernel when parameterized by the marked parameter unless $\NP\subseteq \coNP/poly$; the white one if a polynomial kernel exists. 
The bold borders highlight parameters that are covered in this paper.
}
\label{fig:result-overview}
\end{figure}
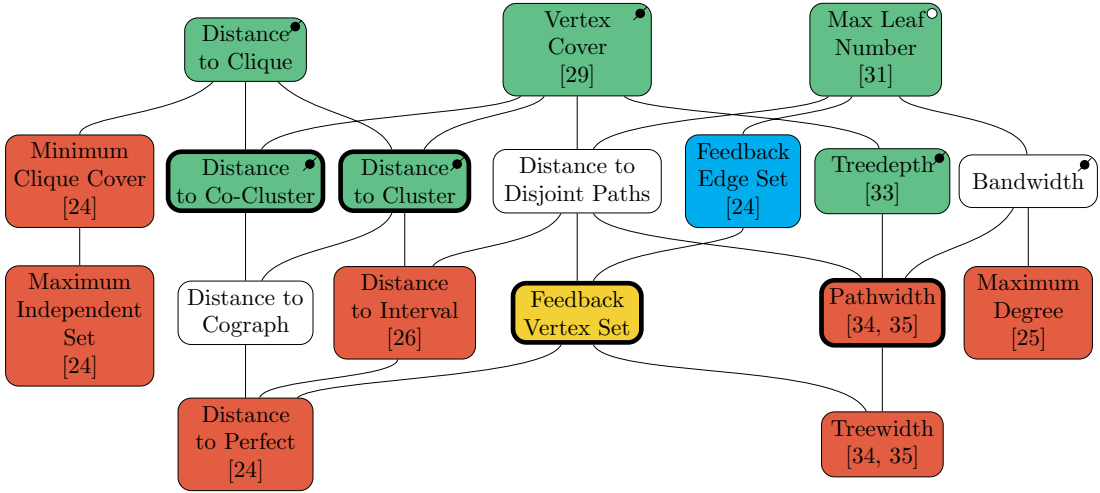

\subparagraph*{Our contributions.} 
In this paper, we continue the analysis of  structural parameterizations of \textsc{Metric Dimension}.
See the Hasse diagram in Figure~\ref{fig:result-overview} for a summary of known results and our contributions.
As mentioned before, it is known that {\sc Metric Dimension} is \W[1]-hard parameterized by the treewidth (in fact, even the pathwidth)~\cite{BP21}.
There are two natural directions to improve this result.
One direction was to show that {\sc Metric Dimension} is para-\NP-hard parameterized by the treewidth, which was proven even for pathwidth in~\cite{LM22}.
Another direction is to prove that {\sc Metric Dimension} is \W[1]-hard for a higher parameter than treewidth, {\it i.e.}, one for which the treewidth is upper bounded by a function of it, such as the pathwidth.
Another parameter fitting this profile is the feedback vertex set number since the treewidth of a graph $G$ is upper bounded by the feedback vertex set number of $G$ plus one.
Moreover, the complexity of {\sc Metric Dimension} parameterized by the feedback vertex set number is left as an open problem in~\cite{HartungN13}, the seminal paper on the parameterized complexity of {\sc Metric Dimension}.\footnote{Note that it is even unknown whether {\sc Metric Dimension} admits an \XP\ algorithm parameterized by the feedback vertex set number.}
We take this direction and answer this open question of~\cite{HartungN13} by proving that {\sc Metric Dimension} is \W[1]-hard parameterized by the combined parameter feedback vertex set number plus pathwidth (see \Cref{sec:fvs}). 
On the positive side, we then show that {\sc Metric Dimension} is \FPT~for the structural parameters the distance to cluster and the distance to co-cluster both of which are smaller parameters than the vertex cover number (see \Cref{sec:FPT}).
Note that the \FPT\ algorithm for the distance to cluster parameter implies an \FPT\ algorithm for the distance to clique parameter.

\section{Preliminaries}
\label{sec:preliminaries}

\subparagraph*{Graph theory.} 
We use standard graph-theoretic notation and refer the reader to~\cite{D12} for any undefined notation. 
For an undirected graph $G$, sets $V(G)$ and $E(G)$ denote its set of vertices and edges, respectively.
Two vertices $u,v\in V(G)$ are {\it adjacent} or {\it neighbors} if $uv\in E(G)$. 
The {\it open neighborhood} of a vertex $u\in V(G)$, denoted by $N(u):=N_G(u)$, is the set of vertices that are neighbors of $u$. 
The {\it closed neighborhood} of a vertex $u\in V(G)$ is denoted by $N[u]:=N_G[u]:=N_G(u)\cup \{u\}$.
For any $X \subseteq V(G)$ and $u\in V(G)$, $N_X(u) = N_G(u) \cap X$. 
Any two vertices $u,v\in V(G)$ are {\it true twins} if $N[u] = N[v]$, and are {\it false twins} if $N(u) = N(v)$. 
Observe that if $u$ and $v$ are true twins, then $uv \in E(G)$, but if they are only false twins, then $uv \not \in E(G)$.
For any $u,v \in V(G)$, we say that $u$ is connected to $v$ by a path $P$ of length $\ell$ if $P= w_0w_1\ldots w_\ell$, where $w_0 = u$ and $v = w_\ell$.
For a path $P$, we denote the length of $P$ by $\mathsf{lgt}(P)$,  and, for any $u,v \in V(P)$, we let $P[u,v]$ be the subpath of $P$ from $u$ to $v$.
The {\it distance} between two vertices $u,v\in V(G)$ in $G$, denoted by $d(u,v):=d_G(u,v)$, is the length of a $(u,v)$-shortest path in $G$. The {\it distance} between two given subsets $X,Y \subseteq V(G)$, denoted by $d(X,Y)$, is the minimum length of a shortest path from a vertex in $X$ to a vertex in $Y$, {\it i.e.}, $d(X,Y) = \min_{x \in X, y\in Y} d(x,y)$.
For a subset $S$ of $V(G)$, we denote the graph obtained by deleting $S$ from $G$ by $G - S$.
We denote the subgraph of $G$ induced on the set $S$ by $G[S]$.
The {\it complement} of a graph $G$ is a graph $H$ with the same vertex set, and such that any two vertices $u,v\in V(G)$ are adjacent in $H$ if and only if they are not adjacent in $G$.

A set of vertices $Y$ is said to be {an} \emph{independent set} if no two vertices in $Y$ are adjacent.
For a graph $G$, a set $X \subseteq V(G)$ is said to be {a} \emph{vertex cover} if $V(G) \setminus X$ is an independent set.
A vertex cover $X$ is a \emph{minimum vertex cover} if for any other vertex cover $Y$ of $G$, we have $|X| \le |Y|$.
\emph{The vertex cover number} of graph $G$ is the size of {a} minimum vertex cover of {a graph} $G$.
For a graph $G$, a set $X \subseteq V(G)$ is said to be {a} \emph{feedback vertex set} if $V(G) \setminus X$ is a acyclic graph.
We define the notation of \emph{the feedback vertex set number} in the analogous way.
A set of vertices $Y$ is said to be a \emph{clique} if any two vertices in $Y$ are adjacent.
We say graph $G$ is a \emph{cluster graph} if it a disjoint union of cliques.
Also, we say $G$ is a \emph{co-cluster graph} if its complement is  a cluster graph.
For a graph class $\calF$, we say a subset $X \subseteq V(G)$ is a \emph{$\calF$-modulator} of $G$ if $G - X \in \calF$.
A $\calF$-modulator $X$ is a \emph{minimum $\calF$-modulator} if for any other $\calF$-modulator $Y$ of $G$, we have $|X| \le |Y|$.
\emph{The distance to $\calF$} of graph $G$ is the size of a minimum $\calF$-modulator.

\subparagraph*{Metric Dimension.}
A subset of vertices $S\subseteq V(G)$ {\it resolves} a pair of vertices $u,v\in V(G)$ if there exists a vertex $w \in S$ such that $d(w,u)\neq d(w,v)$.
A vertex $u\in V(G)$ is {\it distinguished} by a subset of vertices $S\subseteq V(G)$ if, for any $v\in V(G)\setminus \{u\}$, there exists a vertex $w\in S$ such that $d(w,u)\neq d(w,v)$.
For an ordered subset of vertices $S=\{s_1,\dots,s_k\}\subseteq V(G)$ and a single vertex $u\in V(G)$, the {\it distance vector} of $S$ with respect to $u$ is $r(S|u):=(d(s_1,u),\dots,d(s_k,u))$.
The following observation is used throughout the paper.

\begin{observation}\label{obs:twins}
Let $G$ be a graph. 
Then, for any (true or false) twins $u,v \in V(G)$ and any $w \in V(G) \setminus \{u,v\}$, $d(u,w) = d(v,w)$, 
and so, for any resolving set $S$ of $G$, $S\cap \{u,v\} \neq \emptyset$.
\end{observation}

\begin{proof}
As $w\in V(G) \setminus \{u,v\}$, and $u$ and $v$ are (true or false) twins, the shortest $(u,w)$- and $(v,w)$-paths contain a vertex of $N:=N(u)\setminus\{v\}=N(v)\setminus\{u\}$. 
Thus, $d(u,w)=\min\limits_{s\in S}d(s,w)+1=d(v,w)$, and so, any resolving set $S$ of $G$ contains at least one of $u$ and $v$.
\end{proof}

\subparagraph*{Parameterized Complexity.}
An instance of a parameterized problem $\Pi$ comprises an input $I$, which is an input of the classical instance of the problem and an integer $\ell$, which is called as the parameter.
A problem $\Pi$ is said to be \emph{fixed-parameter tractable} or in \FPT\ if given an instance $(I,\ell)$ of $\Pi$, we can decide whether or not $(I,\ell)$ is a \yes-instance of $\Pi$ in  time $f(\ell)\cdot |I|^{\OO(1)}$.
for some computable function $f$ whose value depends only on $\ell$. 

A \emph{compression} of a parameterized problem $\Pi_1$ into a (non-parameterized) problem $\Pi_2$ is a polynomial algorithm that maps each instance $(I_1, \ell_1)$ of $\Pi_1$ to an instance $I$ of $\Pi_2$ such that $(1)$ $(I, \ell)$ is a \yes-instance of $\Pi_1$ if and only if $I_2$ is a \yes-instance of $\Pi_2$, and $(2)$ the size of  $I_2$ is bounded by $g(\ell)$ for a computable function $g$.
The output $I_2$ is also called a \emph{compression} and
the function $g$ is said to be the size of the compression.
A compression is polynomial if $g$ is polynomial.
A compression is said to be \emph{kernel} if $\Pi_1 = \Pi_2$.
It is known that the problem is \FPT\ if and only if it admits a kernel (See, for example, \cite[Lemma 2.2]{DBLP:books/sp/CyganFKLMPPS15}). 

It is typical to describe a compression or kernelization algorithm as a series of reduction rules.
A \emph{reduction rule} is a polynomial time algorithm that takes as an input an instance of a problem and outputs another (usually reduced) instance.
A reduction rule said to be \emph{applicable} on an instance if the output instance is different from the input instance.
A reduction rule is \emph{safe} if the input instance is a \yes-instance if and only if the output instance is a \yes-instance.
For more on parameterized complexity and related terminologies, we refer the reader to the recent book by Cygan et al.~\cite{DBLP:books/sp/CyganFKLMPPS15}.
\section{Feedback Vertex Set Number plus Pathwidth}
\label{sec:fvs}

In this section, we prove that {\sc Metric Dimension} is $\W[1]$-hard parameterized by the combined parameter feedback vertex set number plus pathwidth.
To prove this, we reduce from the {\sc NAE-Integer-3-Sat} problem defined as follows. 
An instance of this problem consists of a set $X$ of variables, a set $C$ of clauses, and an integer $d$.
Each variable takes a value in $\{1,\ldots,d\}$, 
and clauses are of the form $(x \leq a_x, y \leq a_y, z \leq a_z)$, where $a_x,a_y,a_z \in \{1,\ldots,d\}$. 
A clause is satisfied if not all three inequalities are true and not all are false.
The goal is to find an assignment of the variables that satisfies all given clauses.
This problem was shown to be $\W[1]$-hard parameterized by the number of variables~\cite{stacs2015}.
The remainder of this section is devoted to the proof of the following.

\begin{theorem}
\label{thm:fvs}
{\sc Metric Dimension} is $\W[1]$-hard parameterized by the combined parameter feedback vertex set number plus pathwidth.
\end{theorem}

We present the reduction in \Cref{subsec:reduction}, prove preliminary claims in \Cref{subsec:prelim_clm}, and show the correctness of the reduction in \Cref{subsec:correctness}.

\subsection{Reduction}
\label{subsec:reduction}

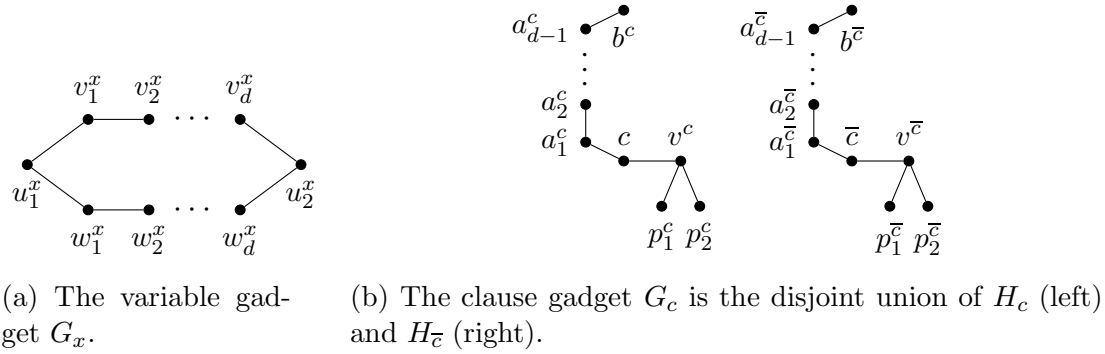
\begin{figure}
\centering
\begin{subfigure}[b]{.25\textwidth}
\centering
\begin{tikzpicture}[scale=.8]
\node[circ,label=below:{\small $u^x_1$}] (u1) at (0,.75) {};

\node[circ,label=above:{\small $v^x_1$}] (v1) at (1,1.5) {};
\node[circ,label=above:{\small $v^x_2$}] (v2) at (2,1.5) {};
\node[draw=none] at (2.75,1.5) {$\cdots$};
\node[circ,label=above:{\small $v^x_d$}] (vd) at (3.5,1.5) {};

\node[circ,label=below:{\small $w^x_1$}] (w1) at (1,0) {};
\node[circ,label=below:{\small $w^x_2$}] (w2) at (2,0) {};
\node[draw=none] at (2.75,0) {$\cdots$};
\node[circ,label=below:{\small $w^x_d$}] (wd) at (3.5,0) {};

\node[circ,label=below:{\small $u^x_2$}] (u2) at (4.5,.75) {};

\draw (u1) -- (v1)
(u1) -- (w1)
(v1) -- (v2)
(w1) -- (w2)
(vd) -- (u2)
(wd) -- (u2); 
\end{tikzpicture}
\caption{The variable gadget $G_x$.}
\label{fig:vargad}
\end{subfigure}
\hspace*{.5cm}
\begin{subfigure}[b]{.65\textwidth}
\centering
\begin{tikzpicture}
\node[circ,label=above:{\small $c$}] (c1) at (0,.75) {};
\node[circ,label=above:{\small $v^c$}] (v) at (.75,.75) {};
\node[circ,label=below:{\small $p_1^c$}] (t1) at (.5,.15) {};
\node[circ,label=below:{\small $p_2^c$}] (t2) at (1,.15) {};
\node[circ,label=left:{\small $a^c_1$}] (a1) at (-.5,1) {};
\node[circ,label=left:{\small $a^c_2$}] (a2) at (-.5,1.5) {};
\node[circ,label=left:{\small $a^c_{d-1}$}] (ad) at (-.5,2.5) {};
\node[draw=none,rotate=90] at (-.5,2) {$\cdots$};
\node[circ,label=below:{\small $b^c$}] (bc) at (0,2.75) {};

\draw (c1) -- (v)
(v) -- (t1)
(v) -- (t2)
(c1) -- (a1)
(a1)-- (a2)
(ad) -- (bc);

\node[circ,label=above:{\small $\overline{c}$}] (c1b) at (3,.75) {};
\node[circ,label=above:{\small $v^{\overline{c}}$}] (vb) at (3.75,.75) {};
\node[circ,label=below:{\small $p_1^{\overline{c}}$}] (t1b) at (3.5,.15) {};
\node[circ,label=below:{\small $p_2^{\overline{c}}$}] (t2b) at (4,.15) {};
\node[circ,label=left:{\small $a^{\overline{c}}_1$}] (a1b) at (2.5,1) {};
\node[circ,label=left:{\small $a^{\overline{c}}_2$}] (a2b) at (2.5,1.5) {};
\node[circ,label=left:{\small $a^{\overline{c}}_{d-1}$}] (adb) at (2.5,2.5) {};
\node[draw=none,rotate=90] at (2.5,2) {$\cdots$};
\node[circ,label=below:{\small $b^{\overline{c}}$}] (bcb) at (3,2.75) {};

\draw (c1b) -- (vb)
(vb) -- (t1b)
(vb) -- (t2b)
(c1b) -- (a1b)
(a1b)-- (a2b)
(adb) -- (bcb);
\end{tikzpicture}
\caption{The clause gadget $G_c$ is the disjoint union of $H_c$ (left) and $H_{\overline{c}}$ (right).}
\label{fig:clausegad}
\end{subfigure}
\caption{The gadgets in the proof of \Cref{thm:fvs}.}
\end{figure}

We reduce from {\sc NAE-Integer-3-Sat}: 
given an instance $(X,C,\allowbreak d)$ of this problem, we construct an instance $(G,k)$ of {\sc Metric Dimension} as follows. 
\begin{itemize}
\item For each variable $x \in X$, we introduce a cycle $G_x$ of length $2d+2$ 
which has two distinguished \emph{anchor vertices} $u^x_1$ and $u^x_2$ as depicted in \Cref{fig:vargad};
for convenience, we may also refer to $u_1^x$ as $v^x_0$ or $w^x_0$, and to $u^x_2$ as $v^x_{d+1}$ or $w^x_{d+1}$.
\item For each clause $c = (x \leq a_x, y \leq a_y, z \leq a_z)$, we introduce the gadget $G_c$ depicted in \Cref{fig:clausegad}
consisting of two vertex-disjoint copies $H_c$ and $H_{\overline{c}}$ of the same graph. 
More precisely, for $\ell \in \{c,\overline{c}\}$, $H_\ell$ consists of a $K_{1,3}$ on the vertex set $\{\ell,v^\ell,p^\ell_1,p^\ell_2\}$, 
where $v^\ell$ has degree three,
and a path $P_{b^\ell}$ of length $d$ connects $\ell$ to $b^\ell$. 

The subgraph of $G_c$ induced by $\{\ell,v^\ell,p^\ell_1,p^\ell_2~|~\ell \in \{c,\overline{c}\}\}$ is referred to as the \emph{core of $G_c$}.

\item We further connect $G_c$ to $G_x,G_y$, and $G_z$ as follows.
\begin{itemize}
\item For every $t \in \{x,y,z\}$, we connect $b^c$ to $u^t_1$ by a path $P^{t,c}_1$ of length $4d-a_t$, and $v^c$ to $u^t_2$ by a path $P^{t,c}_2$ of length $4d+a_t-1$. 
Furthermore, letting $w^{t,c}$ be the neighbor of $v^c$ on $P^{t,c}_2$, 
we attach a copy $W^{t,c}$ of $K_{1,3}$ to $w^{t,c}$ by identifying $w^{t,c}$ with one of the leaves.

We denote by $t^{t,c}_1$ and $t^{t,c}_2$ the two remaining leaves and refer to $W^{t,c}$ as a \emph{pendant claw}.

\item Similarly, for every $t \in \{x,y,z\}$, we connect $b^{\overline{c}}$ to $u^t_2$ by a path $P^{t,\overline{c}}_2$ of length $3d+a_t$, and $v^{\overline{c}}$ to $u^t_1$ by a path $P^{t,\overline{c}}_1$ of length $5d-a_t$.
Furthermore, letting $w^{t,{\overline{c}}}$ be the neighbor of $v^{\overline{c}}$ on $P^{t,\overline{c}}_1$, 
we attach a copy $W^{t,\overline{c}}$ of $K_{1,3}$ to $w^{t,{\overline{c}}}$ by identifying $w^{t,{\overline{c}}}$ with one of the leaves.

We denote by $t^{t,\overline{c}}_1$ and $t^{t,\overline{c}}_2$ the two remaining leaves 
and refer to $W^{t,\overline{c}}$ as a \emph{pendant claw}.
\end{itemize}

\item Finally, we introduce a 3-vertex path $P = t_1pt_2$ which we connect to the clause gadgets as follows.
\begin{itemize}
\item For every clause $c \in C$ and $\ell \in \{c,\overline{c}\}$, we connect $p$ to $v^\ell$ by a path $P_\ell$ of length $2d$.
\item Furthermore, letting $w^\ell$ be the neighbor of $p$ on $P_\ell$,
we attach a copy $W^\ell$ of $K_{1,3}$ to $w^\ell$ by identifying $w^\ell$ with one of the leaves.

We denote by $t^\ell_1$ and $t^\ell_2$ the two remaining leaves and refer to $W^\ell$ as a \emph{pendant claw}.
\end{itemize}
\end{itemize}
This concludes the construction of $G$ (see \Cref{fig:fvsred}).
We set $k = |X| + 10|C| + 1$, and return $(G, k)$ as an instance of \textsc{Metric Dimension}.

Observe that the feedback vertex set number of $G$ is at most $2|X|+1$:
indeed, removing $\{p\} \cup\{u_1^x,u_2^x~|~x \in X\}$ from $G$ results in a graph without cycles.
Furthermore, observe that the pathwidth of $G$ is at most $2|X|+\alpha$ for some constant $\alpha>0$:
indeed, after removing the feedback vertex set above (it can be included in each bag of the path decomposition), each of the remaining connected components is a tree with a constant number of vertices of degree greater than $2$.
Thus, each of these connected components can be dealt with sequentially along the path decomposition, and, for each such connected component, the constant number of vertices of degree greater than $2$ that it contains can be included in all of the bags dedicated to it (hence, decomposing these connected components into disjoint paths which have pathwidth~$1$).

The main ideas behind the reduction are as follows. Using Observation~\ref{obs:twins}, the size of $k$, and other properties of resolving sets in graphs, it can be deduced that, for all $x\in X$, exactly one vertex from $\{v_1^x,\ldots,v_d^x,w_1^x,\ldots,w_d^x\}$ must be in any resolving set $S$ of $G$ of size $k$, and, for each clause $c$, exactly one vertex from $\{p_1^c,p_2^c\}$ must be in $S$. Then, the crux is that, for any clause $c = (x \leq a_x, y \leq a_y, z \leq a_z)$, the only vertices in $S$ that can distinguish $c$ from whichever of $p_1^c$ and $p_2^c$ is not in $S$, say $p_2^c$, are the ones in $\{v_1^x,\ldots,v_d^x,w_1^x,\ldots,w_d^x\}$, $\{v_1^y,\ldots,v_d^y,w_1^y,\ldots,w_d^y\}$, and $\{v_1^z,\ldots,v_d^z,w_1^z,\ldots,w_d^z\}$ whose numbers in the subscripts satisfy the inequalities $x \leq a_x$, $y \leq a_y$, and $z \leq a_z$, respectively. Indeed, since the path lengths depend on the inequalities, only for those vertices that satisfy an inequality does the shortest path from them to $p_2^c$ pass through $c$ before $p_2^c$ ({\it i.e.}, the shortest path contains $P_1^{x,c}$), and thus, distinguishes them. The same holds for $\overline{c}$ and $p_2^{\overline{c}}$ except that the inequalities need to be unsatisfied. The vertex $p$ and the paths connecting it to the clause gadgets exist to create shortcuts that prevent vertices from a variable cycle from distinguishing $c$ from $p_2^c$ ($\overline{c}$ from $p_2^{\overline{c}}$, resp.) for a clause $c$ that does not contain it. Lastly, the pendant claws are added in order to resolve the remainder of the pairs of vertices.

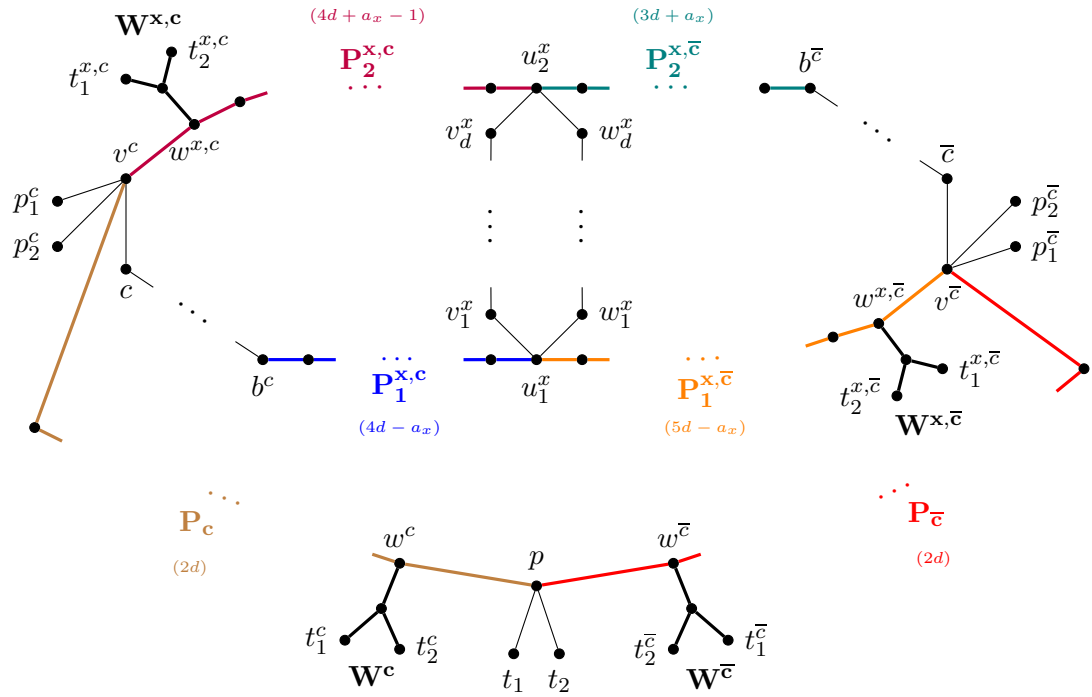
\begin{figure}
\centering
\begin{tikzpicture}[scale=1.2]
\node[circ,label=below:{\small $u_1^x$}] (u1x) at (8,2.5) {};
\node[circ,label=above:{\small $u_2^x$}] (u2x) at (8,5.5) {};
\node[circ,label=left:{\small $v_1^x$}] (v1x) at (7.5,3) {};
\node[circ,label=right:{\small $w_1^x$}] (w1x) at (8.5,3) {};
\node[circ,label=left:{\small $v_d^x$}] (vdx) at (7.5,5) {};
\node[circ,label=right:{\small $w_d^x$}] (wdx) at (8.5,5) {};
\node[draw=none,rotate=90] at (7.5,4) {$\cdots$};
\node[draw=none,rotate=90] at (8.5,4) {$\cdots$};

\draw (u1x) -- (v1x)
(u1x) -- (w1x) 
(vdx) -- (u2x)
(wdx)-- (u2x)
(v1x) -- (7.5,3.3)
(w1x) -- (8.5,3.3)
(vdx) -- (7.5,4.7)
(wdx) -- (8.5,4.7);

\node[circ,label=below:{\small $c$}] (c) at (3.5,3.5) {};
\node[circ,label=above:{\small $v^c$}] (vc) at (3.5,4.5) {};
\node[circ,label=left:{\small $p_1^c$}] (p1c) at (2.75,4.25) {};
\node[circ,label=left:{\small $p_2^c$}] (p2c) at (2.75,3.75) {};
\node[circ,label=below:{\small $w^{x,c}$}] (wxc) at (4.25,5.1) {};
\node[draw=none,rotate=-35] at (4.25,3) {$\cdots$};
\node[circ,label=below:{\small $b^c$}] (bc) at (5,2.5) {};
\node[circ] (pxc1a) at (5.5,2.5) {};
\node[circ] (pxc1b) at (7.5,2.5) {};
\node[draw=none,below,blue] at (6.5,2.5) {\small $\mathbf{P^{x,c}_1}$};
\node[draw=none,blue] at (6.5,2.5) {\small $\cdots$};
\node[draw=none,blue] at (6.5,1.75) {\tiny $(4d-a_x)$};
\node[circ] (c1) at (3.9,5.5) {};
\node[circ,label=left:{\small $t^{x,c}_1$}] (txc1) at (3.5,5.6) {};
\node[circ,label=right:{\small $t^{x,c}_2$}] (txc2) at (4,5.9) {};
\node[draw=none] at (3.75,6.2) {\small $\mathbf{W^{x,c}}$};
\node[circ] (pxc2a) at (4.75,5.35) {};
\node[circ] (pxc2b) at (7.5,5.5) {};
\node[draw=none,above,purple] at (6.15,5.5) {\small $\mathbf{P^{x,c}_2}$}; 
\node[draw=none,rotate=2,purple] at (6.15,5.5) {$\cdots$};
\node[draw=none,purple] at (6.15,6.3) {\tiny $(4d+a_x-1)$};

\draw (c) -- (vc)
(vc) -- (p1c)
(vc) -- (p2c)
(c) -- (3.8,3.3)
(bc) -- (4.7,2.7);

\draw[very thick] (wxc) -- (c1)
(c1) -- (txc1)
(c1) -- (txc2);

\draw[very thick,blue] (bc) -- (pxc1a)
(pxc1a) -- (5.8,2.5)
(u1x) -- (pxc1b)
(pxc1b) -- (7.2,2.5);

\draw[very thick,purple] (vc) -- (wxc)
(wxc) -- (pxc2a)
(pxc2a) -- (5.05,5.45)
(u2x) -- (pxc2b)
(pxc2b) -- (7.2,5.5);

\node[circ,label=below:{\small $v^{\overline{c}}$}] (vcbar) at (12.5,3.5) {};
\node[circ,label=above:{\small $\overline{c}$}] (cbar) at (12.5,4.5) {};
\node[circ,label=right:{\small $p_2^{\overline{c}}$}] (p2cbar) at (13.25,4.25) {};
\node[circ,label=right:{\small $p_1^{\overline{c}}$}] (p1cbar) at (13.25,3.75) {};
\node[draw=none,rotate=-35] at (11.75,5) {$\cdots$};
\node[circ,label=above:{\small $b^{\overline{c}}$}] (bcbar) at (11,5.5) {};
\node[circ] (p2xcbara) at (10.5,5.5) {};
\node[circ] (p2xcbarb) at (8.5,5.5) {};
\node[draw=none,teal] at (9.5,5.5) {$\cdots$};
\node[draw=none,above,teal] at (9.5,5.5) {\small $\mathbf{P_2^{x,\overline{c}}}$};
\node[draw=none,teal] at (9.5,6.3) {\tiny $(3d+a_x)$};
\node[circ,label=above:{\small $w^{x,\overline{c}}$}] (wxcbar) at (11.75,2.9) {};
\node[circ] (p1xcbarba) at (11.25,2.75) {};
\node[circ] (p1xcbarbb) at (8.5,2.5) {};
\node[draw=none,below,orange] at (9.85,2.5) {\small $\mathbf{P^{x,\overline{c}}_1}$};
\node[draw=none,rotate=2,orange] at (9.85,2.5) {$\cdots$};
\node[draw=none,orange] at (9.85,1.75) {\tiny $(5d-a_x)$};
\node[circ] (c2) at (12.05,2.5) {};
\node[circ,label=right:{\small $t^{x,\overline{c}}_1$}] (t1xcbar) at (12.45,2.4) {};
\node[circ,label=left:{\small $t^{x,\overline{c}}_2$}] (t2xcbar) at (11.95,2.1) {};
\node[draw=none] at (12.3,1.8) {\small $\mathbf{W^{x,\overline{c}}}$};

\draw (vcbar) -- (cbar)
(vcbar) -- (p2cbar)
(vcbar) -- (p1cbar) 
(cbar) -- (12.2,4.7)
(bcbar) -- (11.3,5.3);

\draw[very thick] (wxcbar) -- (c2)
(c2) -- (t1xcbar)
(c2) -- (t2xcbar);

\draw[very thick,teal] (bcbar) -- (p2xcbara)
(p2xcbarb) -- (u2x)
(p2xcbarb) -- (8.8,5.5);

\draw[very thick,orange] (vcbar) -- (wxcbar)
(wxcbar) -- (p1xcbarba)
(p1xcbarba) -- (10.95,2.65)
(u1x) -- (p1xcbarbb)
(p1xcbarbb) -- (8.8,2.5);

\node[circ,label=above:{\small $p$}] (p) at (8,0) {};
\node[circ,label=below:{\small $t_1$}] (t1) at (7.75,-.75) {};
\node[circ,label=below:{\small $t_2$}] (t2) at (8.25,-.75) {};
\node[circ,label=above:{\small $w^c$}] (wc) at (6.5,.25) {};
\node[circ] (c3) at (6.3,-.25) {};
\node[draw=none] at (6.2,-1) {\small $\mathbf{W^c}$};
\node[circ,label=left:{\small $t^c_1$}] (tc1) at (5.9,-.6) {};
\node[circ,label=right:{\small $t^c_2$}] (tc2) at (6.5,-.7) {};
\node[circ,label=above:{\small $w^{\overline{c}}$}] (wcbar) at (9.5,.25) {};
\node[circ] (c4) at (9.7,-.25) {};
\node[draw=none] at (9.9,-1) {\small $\mathbf{W^{\overline{c}}}$};
\node[circ,label=right:{\small $t^{\overline{c}}_1$}] (tcbar1) at (10.1,-.6) {};
\node[circ,label=left:{\small $t^{\overline{c}}_2$}] (tcbar2) at (9.5,-.7) {};
\node[circ] (r1) at (2.5,1.75) {};
\node[draw=none,rotate=-20,brown] at (4.6,.95) {$\cdots$}; 
\node[draw=none,below left,brown] at (4.6,.95) {\small $\mathbf{P_c}$};
\node[draw=none,brown] at (4.2,.2) {\tiny $(2d)$};
\node[circ] (r2) at (14,2.4) {};
\node[draw=none,rotate=25,red] at (11.95,1.05) {$\cdots$};
\node[draw=none,below right,red] at (11.95,1.05) {\small $\mathbf{P_{\overline{c}}}$};
\node[draw=none,red] at (12.35,.3) {\tiny $(2d)$};

\draw (p) -- (t1)
(p) -- (t2);

\draw[very thick] (c3) -- (wc)
(c3) -- (tc1)
(c3) -- (tc2)
(wcbar) -- (c4)
(c4) -- (tcbar1)
(c4) -- (tcbar2);  

\draw[very thick,brown] (p) -- (wc)
(wc) -- (6.2,.35)
(r1) -- (vc)
(r1) -- (2.8,1.6);

\draw[very thick,red] (p) -- (wcbar)
(wcbar) -- (9.8,.35)
(r2) -- (vcbar)
(r2) -- (13.7,2.15);
\end{tikzpicture}
\caption{An illustration of the reduction in the proof of \Cref{thm:fvs} with one clause gadget $G_c$ and one variable gadget $G_x$, where $x$ is a variable occurring in $c$ (in between parentheses is the length of the corresponding path).}
\label{fig:fvsred}
\end{figure}

\subsection{Preliminary Claims}
\label{subsec:prelim_clm}

To show that $(G,k)$ is a \yes-instance for {\sc Metric Dimension} if and only if the instance $(X,C,d)$ is satisfiable, we first prove the following claims.

\begin{claim}
\label{clm:distances}
For any two distinct $s,t \in \{c,\overline{c}~|~c \in C\}$ and any two distinct variables $x,y \in X$, the following hold.
\begin{itemize}
\item[(i)] The shortest path from $H_s$ to $H_t$ contains $P_s$ and $P_t$ as subpaths and has length $4d$.
\item[(ii)] $d(V(G_x),V(G_y)) \geq 6d$.
\item[(iii)] If $x$ appears in the clause corresponding to $s$, then $d(V(G_x),V(H_s)) \geq 3d$.
\item[(iv)] If $x$ does not appear in the clause corresponding to $s$,
then any shortest path from $G_x$ to $H_s$ contains $P_s$ as a subpath and has length at least $8d$.
\end{itemize} 
\end{claim}

\begin{claimproof}
Observe first that since any path $P$ from $G_x$ to a clause or another variable gadget contains, as a subpath, some path $P_i^{x,\ell}$, where $i \in [2]$ and $\ell \in \{c,\overline{c}\}$ for some clause $c$ containing $x$, then 
$$\mathsf{lgt}(P) \geq \mathsf{lgt}(P_i^{x,\ell})\geq 3d.$$
This implies, in particular, that
\begin{equation}
\label{eq:dist}
\min_{c \in C} d(V(G_x),V(G_c)) \geq 3d
\end{equation}
and 
$$d(V(G_x),V(G_y)) \geq \min_{c \in C} d(V(G_x),V(G_c)) + \min_{c \in C} d(V(G_y),V(G_c)) \geq 6d.$$
Thus, items (ii) and (iii) hold true.
To prove item (i), let $P$ be a shortest path from $H_s$ to $H_t$.
Observe that $\mathsf{lgt}(P) \leq 4d$ since $P_s$ together with $P_t$ form a path from $H_s$ to $H_t$ of length $4d$.
If $P$ does not contain $P_s$ as a subpath, then there exist $i \in [2]$ and a variable $x$ contained in the clause corresponding to $s$
such that $P$ contains $P^{x,s}_i$ as a subpath.
But then,
$$\mathsf{lgt}(P) \geq \mathsf{lgt}(P^{x,s}_i) + d(V(G_x),V(H_t)) \geq 3d + 3d$$
by \Cref{eq:dist}, a contradiction to the above observation.
Thus, $P$ contains $P_s$ as a subpath
and we conclude, by symmetry, that $P$ also contains $P_t$ as a subpath.
Finally, to prove item (iv), let $P$ be a shortest path from $G_x$ to $H_s$, and let $c \in C$ be a clause containing $x$. 
Then, 
$$\mathsf{lgt}(P) \leq \mathsf{lgt}(P_s) + \mathsf{lgt}(P_c) + d(v^c,V(G_x)) \leq 4d + \min\{d + P_1^{x,c},P_2^{x,c}\} < 4d + 5d.$$
Now, if $P$ does not contain $P_s$ as a subpath, 
then there exist $i \in [2]$ and a variable $z \in X$ contained in the clause corresponding to $s$ 
such that $P$ contains $P_i^{z,s}$;
in particular, since $z \neq x$ by assumption,
$$\mathsf{lgt}(P) \geq \mathsf{lgt}(P_i^{z,s}) + d(V(G_z),V(G_x)) \geq 3d + 6d,$$
a contradiction to the above.
Thus, $P$ contains $P_s$ as a subpath, and so,
$$\mathsf{lgt}(P) = \mathsf{lgt}(P_s) + d(p,V(G_x)) = 2d + 2d + \min_{\ell \in \{c,\overline{c}~|~ c \in C\}} d(v^\ell,V(G_x)) \geq 4d + 4d,$$
which proves item (iv).
\end{claimproof}

\begin{claim}
\label{clm:positive}
For every clause $c = (x \leq a_x,y\leq a_y, z \leq a_z)$ and every $t \in \{x,y,z\}$, the following hold.
\begin{itemize}
\item[(i)] For every $i \in \{0,\ldots,d+1\}$, if $i \leq a_t$, then the shortest path from $v_i^t$ to $c$ contains $P^{t,c}_1$ as a subpath and has length $5d + i - a_t$. Otherwise, the shortest path from $v_i^t$ to $c$ contains $P^{t,c}_2$ as a subpath and has length $5d+1+a_t-i$.
\item[(ii)] For every $i \in \{0,\ldots,d+1\}$, if $i \leq a_t-1$, then the shortest path from $v_i^t$ to $v^c$ contains $P^{t,c}_1$ as a subpath and has length $5d+1+i-a_t$. Otherwise, the shortest path from $v_i^t$ to $v^c$ contains $P^{t,c}_2$ as a subpath and has length $5d+a_t-i$.
\item[(iii)] For every $i \in \{0,\ldots,d+1\}$, if $i \leq a_t-2$, then the shortest path from $v_i^t$ to $t_1^{t,c}$ contains $P^{t,c}_1$ as a subpath and has length $5d+4+i-a_t$. Otherwise, the shortest path from $v_i^t$ to $t_1^{t,c}$ contains $P^{t,c}_2[u_2^t,w^{t,c}]$ as a subpath and has length $5d+1+a_t-i$.
\end{itemize}
\end{claim}

\begin{claimproof}
Consider $i \in \{0,\ldots,d+1\}$.
Then, 
$$d(v_i^t,c) \leq d(v_i^t,u_1^t) + \mathsf{lgt}(P_1^{t,c}) + d(b^c,c) = i + 4d - a_t + d$$
and 
$$d(v_i^t,c) \leq d(v_i^t,u_2^t) + \mathsf{lgt}(P_2^{t,c}) + 1 = d - i + 1 + 4d + a_t - 1 + 1.$$
Now, if a shortest path $P$ from $v^i_t$ to $c$ does not contain $P_1^{t,c}$ nor $P_2^{t,c}$ as a subpath,
then $P$ necessarily contains, as a subpath, a path $P_j^{t,\ell}$ for some $j \in [2]$ and $\ell \in \{c',\overline{c}',\overline{c}\}$, 
where $c'$ is a clause containing $t$. 
In particular, 
$$\mathsf{lgt}(P) \geq \mathsf{lgt}(P_j^{t,\ell}) + d(V(H_\ell),V(H_c)) \geq 3d + 3d.$$
But, $\min \{5d + i - a_t,5d+1+a_t-i\} < 6d$ as $a_t \in [d]$ and $i \in \{0,\ldots,d+1\}$, a contradiction to the above.
Thus, any shortest path from $v^i_t$ to $c$ contains $P_1^{t,c}$ or $P_2^{t,c}$ as a subpath, and so,
$$d(v_i^t,c) = \min \{5d + i - a_t,5d+1+a_t-i\}.$$
Hence, since $5d + i - a_t \leq 5d+1+a_t-i$ if and only if $i \leq a_t + 1/2$, item (i) follows.
Now, as above, we may argue that any shortest path from $v_i^t$ to $v^c$ contains $P^{t,c}_1$ or $P^{t,c}_2$ as a subpath, and so,
\begin{equation*}
\begin{split}
d(v_i^t,v^c) &= \min \{d(v_i^t,u_1^t) + \mathsf{lgt}(P_1^{t,c}) + d(b^c,v^c),d(v_i^t,u_2^t) + \mathsf{lgt}(P_2^{t,c})\}\\
&= \min \{i + 4d - a_t + d+1, d-i+1 + 4d+a_t-1\}.
\end{split}
\end{equation*}
Hence, since $5d+1+i-a_t \leq 5d+a_t-i$ if and only if $i \leq a_t - 1/2$, item (ii) follows.
Finally, by arguing as above, 
we can show that any shortest path from $v_i^t$ to $t^{t,c}_1$ contains $P^{t,c}_1$ or $P^{t,c}_2[u_2^t,w^{t,c}]$ as a subpath, and so,
\begin{equation*}
\begin{split}
d(v_i^t,t^{t,c}_1) &= \min \{d(v_i^t,u_1^t) + \mathsf{lgt}(P_1^{t,c}) + d(b^c,t^{t,c}_1),d(v_i^t,u_2^t) + \mathsf{lgt}(P_2^{t,c})[u_2^t,w^{t,c}] + 2\}\\
&= \min \{i + 4d - a_t + d +4, d- i + 1 + 4d + a_t - 2 + 2\}
\end{split}
\end{equation*}
Hence, since $5d + 4 + i - a_t \leq 5d + 1 + a_t -i$ if and only if $i \leq a_t - 3/2$, item (iii) follows.
\end{claimproof}

\begin{claim}
\label{clm:negative}
For every clause $c = (x \leq a_x,y\leq a_y, z \leq a_z)$ and every $t \in \{x,y,z\}$, the following hold.
\begin{itemize}
\item[(i)] For every $i \in \{0,\ldots,d+1\}$, if $i \leq a_t$, then the shortest path from $v_i^t$ to $\overline{c}$ contains $P^{t,\overline{c}}_1$ as a subpath and has length $5d+1+i-a_t$. Otherwise, the shortest path from $v_i^t$ to $\overline{c}$ contains $P^{t,\overline{c}}_2$ as a subpath and has length $5d+1+a_t-i$.
\item[(ii)] For every $i \in \{0,\ldots,d+1\}$, if $i \leq a_t+1$, then the shortest path from $v_i^t$ to $v^{\overline{c}}$ contains $P^{t,\overline{c}}_1$ as a subpath and has length $5d + i - a_t$. Otherwise, the shortest path from $v_i^t$ to $v^{\overline{c}}$ contains $P^{t,\overline{c}}_2$ as a subpath and has length $5d+2+a_t-i$.
\item[(iii)] For every $i \in \{0,\ldots,d+1\}$, if $i \leq a_t+2$, then the shortest path from $v_i^t$ to  $t_1^{t,\overline{c}}$ contains $P^{t,\overline{c}}_1[u_1^t,w^{t,\overline{c}}]$ as a subpath and has length $5d+1+i-a_t$. Otherwise, the shortest path from $v_i^t$ to $t_1^{t,\overline{c}}$ contains $P^{t,\overline{c}}_2$ as a subpath and has length $5d+5+a_t-i$.
\end{itemize}
\end{claim} 

\begin{claimproof}
As in the proof of \Cref{clm:positive}, we may argue that for every $i \in \{0,\ldots,d+1\}$, 
any shortest path from $v_i^t$ to $\overline{c}$ (or $v^{\overline{c}}$) contains $P^{t,\overline{c}}_1$ or $P^{t,\overline{c}}_2$ as a subpath.
Similarly, any shortest path from $v_i^t$ to $t_1^{t,\overline{c}}$ contains 
$P^{t,\overline{c}}_1[u_1^t,w^{t,\overline{c}}]$ or $P^{t,\overline{c}}_2$ as a subpath.
It follows that 
\begin{equation*}
\begin{split}
d(v_i^t,\overline{c}) &=\min \{d(v_i^t,u_1^t) + \mathsf{lgt}(P^{t,\overline{c}}_1) + 1,
d(v_i^t,u_2^t) + \mathsf{lgt}(P^{t,\overline{c}}_2) + \mathsf{lgt}(P_{b^{\overline{c}}})\}\\
&= \min\{i + 5d - a_t + 1, d-i+1+3d+a_t+d\}.
\end{split}
\end{equation*}
Hence, since $5d+1+i-a_t \leq 5d + 1+a_t - i$ if and only if $i \leq a_t$, item (i) follows.
Similarly, 
\begin{equation*}
\begin{split}
d(v_i^t,v^{\overline{c}}) &=\min \{d(v_i^t,u_1^t) + \mathsf{lgt}(P^{t,\overline{c}}_1),
d(v_i^t,u_2^t) + \mathsf{lgt}(P^{t,\overline{c}}_2) + \mathsf{lgt}(P_{b^{\overline{c}}}) +1\}\\
&= \min\{i + 5d - a_t, d-i+1+3d+a_t+d+1\}.
\end{split}
\end{equation*}
Hence, since $5d+i-a_t \leq 5d+2+a_t -i$ if and only if $i \leq a_t+1$, item (ii) follows.
Finally, letting $\widetilde{P}_2^{t,\overline{c}}$ be the path obtained by concatenating $P_2^{t,\overline{c}}$ and $P_{b^{\overline{c}}}$ (so that the next equation fits in one line), then
\begin{equation*}
\begin{split}
d(v_i^t,t^{t,\overline{c}}_1)& = \min \{d(v^t_i,u_1^t) + \mathsf{lgt}(P_1^{t,\overline{c}}[u_1^t,w^{t,\overline{c}}]) + d(w^{t,\overline{c}},t^{t,\overline{c}}_1), d(v^t_i,u_2^t) + \mathsf{lgt}(\widetilde{P}_2^{t,\overline{c}}) + d(\overline{c},t^{t,\overline{c}}_1)\}\\
&= \min \{i + 5d - a_t - 1 + 2, d - i + 1 + 3d + a_t +d + 4\}.
\end{split}
\end{equation*}
Hence, since $5d + 1 + i - a_t \leq  5d + 5 + a_t - i$ if and only if $i \leq a_t + 2$, item (iii) follows.
\end{claimproof}

\subsection{Correctness of the reduction}
\label{subsec:correctness}

The aim of this section is to prove that $(G,k)$ is a \yes-instance for \mdfull\ if and only if $(X,C,d)$ is a \yes-instance for {\sc NAE-Integer-3-Sat}. We prove the two directions in two separate lemmas.

\begin{lemma}
If $(G,k)$ is a \yes-instance for \mdfull\ then $(X,C,d)$ is a \yes-instance for {\sc NAE-Integer-3-Sat}.
\end{lemma}

\begin{proof}
Assume that $(G,k)$ is a \yes-instance for {\sc Metric Dimension} and let $S$ be a resolving set of size at most $k$. 
By Observation~\ref{obs:twins}, for any clause $c \in C$,
\begin{equation}
\label{eq1}
|S \cap \{p_1^c,p_2^c\}| \geq 1,~|S \cap \{p_1^{\overline{c}},p_2^{\overline{c}}\}| \geq 1,~|S \cap \{t_1^c,t_2^c\}| \geq 1,\text{ and }|S \cap \{t_1^{\overline{c}},t_2^{\overline{c}}\}| \geq 1.
\end{equation}
Also, by Observation~\ref{obs:twins}, for any clause $c \in C$ and any variable $x \in X$ appearing in $c$,
\begin{equation}
\label{eq2}
|S \cap \{t^{x,c}_1,t^{x,c}_2\}| \geq 1 \text{ and } |S \cap \{t^{x,\overline{c}}_1,t^{x,\overline{c}}_2\}| \geq 1.
\end{equation}
For the same reason,
\begin{equation}
\label{eq3}
|S \cap  \{t_1,t_2\}| \geq 1.
\end{equation}
Consider now a variable $x \in X$.
Since any path from a vertex in $V(G) \setminus V(G_x)$ to a vertex in $\{v^x_i,w^x_i~|~i \in [d]\}$ contains $u^x_1$ or $u^x_2$, 
and, for any $i \in [d]$ and $u \in \{u^x_1,u^x_2\}$,  $d(u,v^x_i) = d(u,w^x_i)$, 
no vertex in $V(G) \setminus \{v^x_i,w^x_i~|~i \in [d]\}$ can resolve $v^x_i$ and $w^x_i$ for any $i \in [d]$. 
It follows that
\begin{equation}
\label{eq4}
|S \cap \{v^x_i,w^x_i~|~i \in [d]\}| \geq 1.
\end{equation}
Now, note that $S$ has size at most $k=|X|+10|C|+1$, and so, equality must in fact hold in every inequality of \Cref{eq1,eq2,eq3,eq4}. 
Without loss of generality, let us assume that $t_1 \in S$ and that, for every clause $c \in C$ and variable $x \in X$ appearing in $c$, we have that 
$p_1^c,p_1^{\overline{c}},t_1^c,t_1^{\overline{c}},t^{x,c}_1,t^{x,\overline{c}}_1 \in S$.

For every variable $x \in X$, assume, without loss of generality, 
that $S \cap \{v_i^x,w_i^x~|~i \in [d]\} = S \cap \{v_i^x|~i \in [d]\}$, 
and let $i_x \in [d]$ be the index of the vertex in $S \cap \{v_i^x~|~i \in [d]\}$. 
We contend that the assignment which sets each variable $x$ to $i_x$ satisfies every clause in $C$. 
Indeed, consider a clause $c = (x \leq a_x,y\leq a_y,z \leq a_z)$.
We first aim to show that, for every $w \in S \setminus \{V(G_x) \cup V(G_y) \cup V(G_z)\}$ and $\ell \in \{c,\overline{c}\}$, 
$d(w,\ell) = d(w,p_2^\ell)$.
Note that it suffices to show that any shortest path from $w \in S \setminus \{V(G_x) \cup V(G_y) \cup V(G_z)\}$ to $\ell \in \{c,\overline{c}\}$ contains $v^\ell$,  
as then $d(w,\ell) = d(w,v^\ell) + 1 = d(w,p_2^\ell)$.
Now, if $w \in V(G_t)$ for some $t \in \{c',\overline{c'}~|~c' \in C\}$ different from $\ell$, then this readily follows from \Cref{clm:distances}(i);
and if $w \in V(G_t)$ for some $t \in X \setminus \{x,y,z\}$, then this readily follows from \Cref{clm:distances}(iv).
If $w = t_1^{r,q}$ for some $r \in X$ and $q \in \{c',\overline{c'}~|~c' \in C\}$, 
then $d(t_1^{r,q},\ell) = d(t_1^{r,q},v^q) + d(v_q,\ell)$, and so,
by \Cref{clm:distances}(i), any path from $w$ to $\ell$ contains $v^\ell$.
Finally, if $w \in \{t^{c'}_1,t_1^{\overline{c}'}~|~c' \in C\}$, then clearly
any shortest path from $w$ to $\ell$ contains~$v^\ell$.

Since $S$ is a resolving set, it follows that, for every clause $c\in C$, there exist $t,f \in \{x,y,z\}$ such that 
$d(v_{i_t}^t,c) \neq d(v_{i_t}^t,p_2^c)$ and $d(v_{i_f}^f,\overline{c}) \neq d(v_{i_f}^f,p_2^{\overline{c}})$.
Now, by \Cref{clm:positive}(i) and (ii), if $i_t > a_t$, then
\begin{equation*}
d(v_{i_t}^t,c) = 5d + 1 + a_t - i_t = d(v_{i_t}^t,v^c) + 1 = d(v_{i_t}^t,p_2^c),
\end{equation*}
a contradiction to our assumption.
Therefore, $i_t \leq a_t$.
Similarly, if $i_f \leq a_f$, then by \Cref{clm:negative}(i) and (ii),
\begin{equation*}
d(v_{i_f}^f,\overline{c}) = 5d + 1 + i_f - a_f = d(v_{i_f}^f,v^{\overline{c}}) + 1 = d(v_{i_f}^f,p_2^{\overline{c}}),
\end{equation*}
a contradiction to our assumption.
Therefore, $i_f > a_f$, and so, 
the assignment constructed above indeed satisfies every clause in $C$ which concludes the proof.
\end{proof}

\begin{lemma}
If $(X,C,d)$ is a \yes-instance for {\sc NAE-Integer-3-Sat} then $(G,k)$ is a \yes-instance for \mdfull.
\end{lemma}

\begin{proof}
Assume that $(X,C,d)$ is satisfiable and 
let $\phi: X \to \{1,\ldots,d\}$ be an assignment of the variables satisfying every clause in $C$. 
We construct a resolving set $S$ of $G$ as follows. 
First, we add $t_1$ to $S$. 
For every variable $x \in X$, we add $v_{\phi(x)}^x$ to $S$. 
Finally, for every clause $c \in C$, we add $p_1^c,p_1^{\overline{c}},t_1^c,t_1^{\overline{c}}$ to $S$ 
and further add, for every variable $t$ appearing in $c$, $t_1^{t,c},t_1^{t,\overline{c}}$ to $S$.
Note that $|S|=k$ and that every vertex of $S$ is distinguished by itself.
Let us show that $S$ is indeed a resolving set of $G$.
Consider two distinct vertices $u,v \in V(G)$.
To prove that $S$ resolves the pair $u,v$, we distinguish the following cases.\\

\noindent
\textbf{Case 1.} \emph{At least one of $u$ and $v$ belongs to a pendant claw.}
Without loss of generality, assume first that $u \in V(W^\ell)$, where $\ell \in \{c,\overline{c}~|~c\in C\}$.
If $v \in V(G) \setminus V(W^\ell)$, then $d(t^\ell_1,v) > 2 \geq d(t^\ell_1,u)$.
Suppose therefore that $v \in V(W^\ell)$ as well.
If $\{u,v\} \neq \{w^\ell,t^\ell_2\}$, then $d(t^\ell_1,u) \neq d(t^\ell_1,v)$.
If $\{u,v\} = \{w^\ell,t^\ell_2\}$, then $d(t_1,w^\ell) = 2 < 4 = d(t_1, t^\ell_2)$.
Second, assume that $u \in V(W^{t,\ell})$, where $\ell \in \{c,\overline{c}\}$ for some clause $c\in C$ and $t$ is a variable appearing in clause $c$.
If $v \in V(G) \setminus V(W^{t,\ell})$, then $d(t^{t,\ell}_1,v) > 2 \geq d(t^{t,\ell}_1,u)$.
Suppose therefore that $v \in V(W^{t,\ell})$ as well.
If $\{u,v\} \neq \{w^{t,\ell},t^{t,\ell}_2\}$, then $d(t^{t,\ell}_1,u) \neq d(t^{t,\ell}_1,v)$.
If $\{u,v\} = \{w^{t,\ell},t^{t,\ell}_2\}$, then $d(p^\ell_1,w^{t,\ell}) = 2 < 4 = d(p^\ell_1, t^{t,\ell}_2)$.\\

\noindent
\textbf{Case 2.} \emph{At least one of $u$ and $v$ belongs to the core of a clause gadget.}
Assume, without loss of generality, that $u \in \{\ell,v^\ell,p_1^\ell,p_2^\ell\}$, where $\ell \in \{c,\overline{c}\}$
for some clause $c = (x \leq a_x,y \leq a_y,z \leq a_z)$.
If $v$ is not a neighbor of $v^\ell$, then $d(p_1^\ell,v) > 2 \geq d(p_1^\ell,u)$. 
If $v$ is the neighbor of $v^\ell$ on the path $P_\ell$, then $d(t_1,v) = d < d(t_1,u)$.
Also, $v = w^{t,\ell}$ was covered by the previous case. 
So, consider $v \in \{\ell,v^\ell,p_1^\ell,p_2^\ell\}$.
If $\{u,v\} \neq \{\ell,p_2^\ell\}$, then clearly $d(p_1^\ell,u) \neq d(p_1^\ell,v)$.
Assume therefore that $\{u,v\} = \{\ell,p_2^\ell\}$.
Since $\phi$ satisfies $c$, there exist $t,f \in \{x,y,z\}$ such that $\phi(t) \leq a_t$ and $\phi(f) > a_f$.
Then, either $\phi(t) < a_t$, in which case, by \Cref{clm:positive}(i) and (ii), 
\begin{equation*}
d(v_{\phi(t)}^t,c) = 5d + \phi(t) - a_t < 5d + \phi(t) - a_t + 2 = d(v_{\phi(t)}^t,v^c) + 1 = d(v_{\phi(t)}^t,p_2^c),
\end{equation*}
or $\phi(t) = a_t$, in which case, by \Cref{clm:positive}(i) and (ii),
\begin{equation*}
d(v_{\phi(t)}^t,c) = 5d < 5d + 1 = d(v_{\phi(t)}^t,v^c) + 1 = d(v_{\phi(t)}^t,p_2^c).
\end{equation*}
Similarly, either $\phi(f) = a_f + 1$, in which case, by \Cref{clm:negative}(i) and (ii),
\begin{equation*}
d(v_{\phi(f)}^f,\overline{c}) = 5d < 5d + 2 = d(v_{\phi(f)}^f,v^{\overline{c}}) + 1 = d(v_{\phi(f)}^f,p_2^{\overline{c}}),
\end{equation*}
or $\phi(f) > a_f+1$, in which case, by \Cref{clm:negative}(i) and (ii),
\begin{equation*}
d(v_{\phi(f)}^f,\overline{c}) = 5d + 1+ a_f - \phi(f) < 5d + 3 + a_f - \phi(f) 
= d(v_{\phi(f)}^f,v^{\overline{c}}) + 1 = d(v_{\phi(f)}^f,p_2^{\overline{c}}).
\end{equation*}
In all cases, we conclude that there exists $w \in S$ such that $d(w,\ell) \neq d(w,p_2^\ell)$.\\

\noindent
\textbf{Case 3.} \emph{At least one of $u$ and $v$ belongs to a variable gadget.}
Assume, without loss of generality, that $u \in V(G_x)$ for some variable $x \in X$.
By the previous cases, we may assume that $v$ does not belong to the core of a clause gadget or a pendant claw.
If $v \in V(G_y)$ for some variable $y \neq x$, then by \Cref{clm:distances}(ii), 
\begin{equation*}
d(v^x_{\phi(x)},u) \leq d + 1 < 6d \leq d(V(G_x),V(G_y)) \leq d(v^x_{\phi(x)},v).
\end{equation*}
Now, suppose that $v \in V(G_x)$ as well.
If $\{u,v\} = \{v^x_i,w^x_i\}$ for some $i \in [d]$, then 
\begin{equation*}
d(v^x_{\phi(x)},v^x_i) = |\phi(x) - i| < d(v^x_{\phi(x)},w^x_i) = \min \{\phi(x)+i,2d+2-\phi(x)-i\}.
\end{equation*}
Suppose next that $u = v^x_i$ and $v = v^x_j$ for two distinct $i,j \in \{0,\ldots,d+1\}$, say $i < j$, without loss of generality
Consider a clause $c = (x \leq a_x,y\leq a_y,z \leq a_z)$ containing $x$.
If $j < a_x$, then by \Cref{clm:positive}(ii),
\begin{equation*}
d(p_1^c,v^x_i) = d(v^c,v^x_i) + 1 = 5d + 2 + i - a_x < 5d + 2 + j - a_x = d(v^c,v^x_j) + 1 = d(p_1^c,v^x_j).
\end{equation*}
Now, suppose that $i < a_x \leq j$. 
Then, by \Cref{clm:positive}(ii),
\begin{equation*}
d(p_1^c,v^x_i) - d(p_1^c,v^x_j) = 5d + 2 + i - a_x - (5d + a_x - j + 1) = i + j + 1 - 2a_x.
\end{equation*}
Thus, if $i + j + 1 - 2a_x \neq 0$, then $d(p_1^c,v^x_i) \neq d(p_1^c,v^x_j)$.
Now, if $i+j+1-2a_x=0$, then either $j = a_x$ and $i = a_x -1$, in which case, by \Cref{clm:positive}(iii),
\begin{equation*}
d(t_1^{x,c},v^x_i) = 5d + 2 > 5d + 1 = d(t_1^{x,c},v^x_j),
\end{equation*}
or $j > a_x$ and $i < a_x-1$, in which case, by \Cref{clm:positive}(iii),
\begin{equation*}
d(t_1^{x,c},v^x_j) = 5d + 1 + a_x - j = 5d + 2 + i - a_x < 5d + 4 + i - a_x = d(t_1^{x,c},v^x_i).
\end{equation*}
Finally, if $a_x \leq i < j$, then by \Cref{clm:positive}(ii),
\begin{equation*}
d(p_1^c,v^x_j) = 5d + 1 + a_x - j < 5d + 1 + a_x - i = d(p_1^c,v^x_i).
\end{equation*}
Since for any $t \in V(G) \setminus V(G_x)$ and $k \in [d]$, $d(t,v^x_k) = d(t,w^x_k)$,
we conclude similarly if either $u = v^x_i$ and $v = w^x_j$, or $u = w^x_i$ and $v = w^x_j$ for two distinct $i,j \in \{0,\ldots,d+1\}$.

Assume, henceforth, that $v \notin \bigcup_{x \in X} V(G_x)$.
If $v$ does not belong to a path connecting $G_x$ to some clause gadget, then
\begin{equation*}
d(v^x_{\phi(x)},v) \geq \min_{c \in C} d(V(G_x),V(G_c)) \geq 3d > d+1 \geq d(v^x_{\phi(x)},u)
\end{equation*}
by \Cref{clm:distances}(iii) and (iv).
Suppose therefore that $v \in V(P_i^{x,\ell})$, where $i \in \{1,2\}$ and $\ell \in \{c,\overline{c}\}$
for some clause $c= (x \leq a_x,y\leq a_y,z \leq a_z)$ containing $x$.
Without loss of generality, let us assume that $u = v_j^x$ where $j \in \{0,\ldots,d+1\}$.

Assume first that $\ell = c$ and $i=1$.
Let $P_1^{x,c} = z_0 \ldots z_{4d - a_x}$, where $z_0 = b^c$ and $z_{4d-a_x} = u^x_1$.
Let $v = z_k$, where $k \in [4d - a_x -1]$. 
If $j \leq a_x -1$, then by \Cref{clm:positive}(ii), the shortest path from $p_1^c$ to $v_j^x$ contains $P_1^{x,c}$ as a subpath, 
which implies in particular that $d(p_1^c,v) < d(p_1^c,u)$.
Suppose therefore that $j \geq a_x$.
Then, by \Cref{clm:positive}(ii), 
\begin{equation*}
d(p_1^c,u) - d(p_1^c,v) = 5d+1 + a_x - j - (d+k +2).
\end{equation*}
Thus, if $5d+1 + a_x - j - (d+k +2) \neq 0$, then $d(p_1^c,u) \neq d(p_1^c,v)$.
Now, if $5d+1 + a_x - j - (d+k +2) = 0$, then by \Cref{clm:positive}(iii),
\begin{equation*}
d(t_1^{x,c},u) = 5d+1+a_x-j = d + k + 2 < d + k + 4 = d(t^{x,c}_1,v).
\end{equation*}

Second, assume that $\ell = c$ and $i=2$.
Let $P_2^{x,c} = z_0 \ldots z_{4d + a_x -1}$, where $z_0 = v^c$ and $z_{4d+a_x-1} = u^x_2$.
Let $v = z_k$, where $k \in [4d + a_x -2]$
(note that since $v$ does not belong to the core of a clause gadget or a pendant claw by assumption, in fact $k \geq 2$). 
If $j \geq a_x$, then by \Cref{clm:positive}(ii), the shortest path from $p_1^c$ to $u$ contains $P_2^{x,c}$ as a subpath, 
which implies in particular that $d(p_1^c,v) < d(p_1^c,u)$.
Otherwise, $j \leq a_x -1$, in which case
\begin{equation*}
d(p_1^c,u) - d(p_1^c,v) = 5d+2+j-a_x - (k + 1).
\end{equation*}
Thus, if $5d+2+j-a_x - (k + 1)\neq 0$, then $d(p_1^c,u) \neq d(p_1^c,v)$.
Now, if $5d+2+j-a_x - (k + 1) = 0$, then $j < a_x -1$ since $k < 5d$, and so, by \Cref{clm:positive}(iii),
\begin{equation*}
d(t_1^{x,c},u) = 5d + 4 + j - a_x = k + 3 > k + 1 = d(t^{x,c}_1,v).
\end{equation*}

Third, assume that $\ell = \overline{c}$ and $i =1$.
Let $P_1^{x,\overline{c}} = z_0 \ldots z_{5d - a_x}$, where $z_0 = v^{\overline{c}}$ and $z_{5d-a_x} = u^x_1$.
Let $v = z_k$, where $k \in [5d - a_x -1]$
(note that since $v$ does not belong to the core of a clause gadget or a pendant claw by assumption, in fact $k \geq 2$).
If $j \leq a_x+1$, then by \Cref{clm:negative}(ii), the shortest path from $p_1^{\overline{c}}$ to $v_j^x$ contains 
$P_1^{x,\overline{c}}$ as a subpath which implies in particular that $d(p_1^{\overline{c}},v) < d(p_1^{\overline{c}},u)$.
Suppose therefore that $j \geq a_x + 2$.
Then, by \Cref{clm:negative}(ii),
\begin{equation*}
d(p_1^{\overline{c}},v^x_j) - d(p_1^{\overline{c}},z_k) = 5d + 3 + a_x - j - (k+1).
\end{equation*}
Thus, if $5d + 3 + a_x - j - (k +1) \neq 0$, then $d(p_1^{\overline{c}},v^x_j) \neq d(p_1^{\overline{c}},z_k)$.
Now, if $5d + 3 + a_x - j - (k +1) = 0$, then $j > a_x +2$ since $k < 5d$, and so, by \Cref{clm:negative}(iii),
\begin{equation*}
d(t^{x,\overline{c}}_1,v^x_j) = 5d + 5 + a_x - j = k + 3 > k + 1 = d(t^{x,\overline{c}}_1,z_k).
\end{equation*} 

Assume finally that $\ell = \overline{c}$ and $i = 2$.
Let $P_2^{x,\overline{c}} = z_0 \ldots z_{3d + a_x }$, where $z_0 = b^{\overline{c}}$ and $z_{3d+a_x} = u^x_2$.
Let $v = z_k$, where $k \in [3d + a_x -1]$. 
If $j \geq a_x +2$, then by \Cref{clm:negative}(ii), the shortest path from $p_1^{\overline{c}}$ to $u$ contains 
$P_2^{x,\overline{c}}$ as a subpath, which implies in particular that $d(p_1^{\overline{c}},v) < d(p_1^{\overline{c}},u)$.
Suppose therefore that $j \leq a_x+1$.
Then, by \Cref{clm:negative}(ii),
\begin{equation*}
d(p_1^{\overline{c}},v^x_j) - d(p_1^{\overline{c}},z_k) = 5d +1+j - a_x -(d+k+2).
\end{equation*} 
Thus, if $5d +1+j - a_x -(d+k+2) \neq 0$, then $d(p_1^{\overline{c}},v^x_j) \neq d(p_1^{\overline{c}},z_k)$.
Now, if $5d +1+j - a_x -(d+k+2) = 0$, then $j < a_x + 1$ since $k < 4d$, and so, by \Cref{clm:negative}(iii),
\begin{equation*}
d(t^{x,\overline{c}}_1,v^x_j) = 5d + 1 + j - a_x = d+k + 2 < d+k + 4 = d(t^{x,\overline{c}}_1,z_k).
\end{equation*} 
In all the subcases, we conclude that there exists $w \in S$ such that $d(w,u) \neq d(w,v)$.\\

\noindent
\textbf{Case 4.} \emph{None of the above}.
First, note that $p$ is distinguished by $S$ since it is the unique vertex of $G$ at distance $1$ from $t_1$. Second, $t_2$ is distinguished by $S$ since it is the unique vertex of $G$ at distance $2$ from $t_1$ and distance $4$ from $t_1^{c}$ and $t_1^{\overline{c}}$ for all $c\in C$. Thus, in this last case, we can assume that both $u$ and $v$ belong either to paths connecting gadgets 
or to some path $P_{b^\ell}$, where $\ell \in \{c,\overline{c}~|~c \in C\}$.
Assume first that $u \in V(P_\ell)$ for some $\ell \in \{c,\overline{c}~|~c \in C\}$.
If $v \in V(P_\ell)$ as well, then surely $d(t_1,u) \neq d(t_1,v)$.
If $v \in V(P_q)$ for some $q \in \{c,\overline{c}~|~c \in C\}$ different from $\ell$,
then $d(p_1^\ell,u) < d(p_1^\ell,v)$ since the unique shortest path from $p_1^\ell$ to $v$ contains $P_\ell$ as a subpath.
Finally, if there exists $q \in \{c,\overline{c}~|~c \in C\}$ such that
$v$ belongs to $P_{b_q}$ or to some path connecting $H_q$ to a variable gadget,
then $d(t_1,v) > d(t_1,v^q) \geq d(t_1,u)$. 

Second, assume that $u \in V(P_i^{x,\ell})$, where $i \in [2]$ and $\ell \in \{c,\overline{c}\}$ for some clause $c$ containing variable $x$.
Note that by the previous paragraph, we may assume that $v \notin \bigcup_{q \in C} V(P_q) \cup V(P_{\overline{q}})$.
Suppose first that $v \in V(P_j^{y,q})$, where $j \in [2]$ and $q \in \{c',\overline{c}'\}$ for some clause $c'$ containing variable $y$.
Note that  $d(t^\ell_1,u) = d(t_1,u)$.
So, if $q \neq \ell$, then either $d(t_1,u) \neq d(t_1,v)$, or
\begin{equation*}
d(t_1^\ell,v)- d(t^\ell_1,u) = d(t_1^\ell,p) + d(t_1,v) - 1 -d(t^\ell_1,u) = d(t_1,v) + 2 - d(t_1,u) = 2.
\end{equation*}
Thus, assume that $q = \ell$.
Suppose first that $x = y$.
If $i = j$, then surely $d(p_1^\ell, u) \neq d(p^\ell_1,v)$.
Otherwise, assume, without loss of generality, that $u$ belongs to the path containing $w^{x,\ell}$.
Note that  $d(t^{x,\ell}_1,u) = d(p_1^\ell,u)$.
Then, either $d(p_1^\ell,u) \neq d(p_1^\ell,v)$, or 
\begin{equation*}
d(t^{x,\ell}_1,v) - d(t^{x,\ell}_1,u) = d(t^{x,\ell}_1,v^\ell) + d(p_1^\ell,v) - 1 - d(t^{x,\ell}_1,u) = d(p_1^\ell,v) + 2 - d(p_1^\ell,u) = 2.
\end{equation*} 
Second, suppose that $x \neq y$.
If $u$ belongs to the path containing $w^{x,\ell}$, then we argue as previously.
By symmetry, we may also assume that $v$ does not belong to the path containing $w^{y,\ell}$.
This implies, in particular, that $i = j$ and $b^\ell$ is the endpoint in $H_\ell$ of both $P_i^{x,\ell}$ and $P_j^{y,\ell}$.
Thus, $d(v^x_{\phi(x)},u)\leq d(v^x_{\phi(x)},u^x_i) + d(u_i^x,b^\ell)\leq  d+4d$.
First, note that if a shortest path $P$ from $v^x_{\phi(x)}$ to $v$ contains $P_i^{x,\ell}$ as a subpath, then since $u \in V(P_i^{x,\ell})$, it follows that $d(v^x_{\phi(x)},u) < d(v^x_{\phi(x)},v)$.
Hence, we may assume that $P$ contains a vertex in $G_y$ or both $v^{\ell}$ and $b^\ell$.
By \Cref{clm:distances}(ii), if $P$ contains a vertex in $G_y$, then
\begin{equation*}
d(v^x_{\phi(x)},v)>d(V(G_x), V(G_y))\geq 6d > 5d \geq d(v^x_{\phi(x)},u).
\end{equation*}
Otherwise, $P$ contains $v^\ell$ and $b^\ell$, and so, letting $t \in [2] \setminus \{i\}$, we get that
\begin{equation*}
\mathsf{lgt}(P) \geq d(v^x_{\phi(x)},u^x_t) + d(u_t^x,v^\ell) + d(v^\ell,b^\ell) + d(b^\ell,v)
 \geq 1 + 4d + d + 1 > 5d \geq d(v^x_{\phi(x)},u).
\end{equation*}
Suppose finally that $u \in V(P_{b^\ell})$ for some $\ell \in \{c,\overline{c}~|~c \in C\}$.
By the two previous paragraphs, we may assume that $v \in V(P_{b^q})$ for some $q \in \{c,\overline{c}~|~c \in C\}$.
If $q = \ell$, then surely $d(p_1^\ell,u) \neq d(p_1^\ell,v)$. 
Otherwise, by \Cref{clm:distances}(i),
\begin{equation*}
d(p_1^\ell,v) \geq d(V(H_\ell),V(H_q)) = 4d > d+1 \geq d(p_1^\ell,u)
\end{equation*}
which concludes case 4.\\

By the above case analysis, we conclude that, for any $u,v \in V(G)$, there exists $w \in S$ such that $d(w,u) \neq d(w,v)$,
that is, $S$ is a resolving set of $G$. 
Since $|S| = k$, the lemma follows.
\end{proof}

\section{Distance to Cluster and Distance to co-cluster}\label{sec:FPT}

In this section, we prove that {\sc Metric Dimension} is \FPT\ parameterized by either the distance to cluster or the distance to co-cluster.
In fact, we show that the problem admits an exponential kernel parameterized by the distance to cluster (or co-cluster).
Since the main ideas for these two parameters are the same,
we focus on the distance to cluster parameter.
We remark that applying Reduction Rule~\ref{rr:twin} for false twins (instead of true twins) and defining equivalence classes over the independent sets (instead of cliques) for Reduction Rule~\ref{rr:identical-cliques}, we indeed get the similar result for the distance to co-cluster 
(we provide more details on how to adapt the proofs for the distance to co-cluster at the end of this section).
Recall that, for a graph $G$, the {\it distance to cluster} of $G$ is the minimum number of vertices of $G$ that need to be deleted 
so that the resulting graph is a {\it cluster graph}, that is, a disjoint union of cliques. 

\begin{theorem}
{\sc Metric Dimension} is \FPT\ parameterized by the distance to cluster.
\end{theorem}

\begin{proof}
Let $(G,k)$ be an instance of {\sc Metric Dimension} and
let $X \subseteq V(G)$ be such that $G-X$ is a disjoint union of cliques.
To obtain a kernel for the problem, we present a set of reduction rules.
The safeness of the following reduction rule is trivial.

\begin{reduction rule}
If $V(G) \neq \emptyset$ and $k \le 0$ then return a trivial \no-instance.
\end{reduction rule}

\begin{reduction rule}
\label{rr:twin}
If there exist $u, v, w \in V(G)$ such that $u,v,w$ are true (or false) twins, 
then remove $u$ from $G$ and decrease $k$ by one.
\end{reduction rule}

\begin{lemma}\label{lem:rr:twin-safe}
\Cref{rr:twin} is safe.
\end{lemma}

\begin{proof}
Assume that there exist $u, v, w \in V(G)$ such that $u,v,w$ are (true or false) twins.
We first claim that there exists a minimum resolving set of $G$ that contains $u$.
Indeed, let $S$ be a minimum resolving set of $G$
and suppose that $u \notin S$.
Then, by \Cref{obs:twins}, $v,w \in S$.
Let $S_u = (S \setminus \{v\}) \cup \{u\}$.
We contend that $S_u$ is also a resolving set of $G$.
Indeed, towards a contradiction, suppose that there exist $x,y \in V(G)$ such that no vertex in $S_u$ resolves $x$ and $y$.
Since $S$ is a resolving set of $G$, it must then be that $v$ resolves $x$ and $y$.
But, by \Cref{obs:twins}, $\dist(u,x) = \dist(v,x) \neq \dist(v,y) = \dist(u,y)$, and so, $u$ resolves $x$ and $y$, a contradiction to our assumption.
Thus, $S_u$ is indeed a resolving set. 

We next show that $G$ has a resolving set of size at most $k$ if and only if $G-\{u\}$ has a resolving set of size at most $k-1$.
Assume first that $G$ has a resolving set $S$ of size at most $k$.
By the above paragraph, we may assume that $u \in S$.
Furthermore, by \Cref{obs:twins}, at least of one $v$ and $w$ belongs to $S$ as well, say $v \in S$ without loss of generality
But then, $S \setminus \{u\}$ is a resolving set of $G- \{u\}$:
indeed, for any $x,y \in V(G)$ such that $\dist(u,x) \neq \dist(u,y)$,
we have that $\dist(v,x) \neq \dist(v,y)$ by \Cref{obs:twins}. 
Conversely, if $G-\{u\}$ has a resolving set $S$ of size at most $k-1$,
then it is not difficult to see that $S \cup \{u\}$ is a resolving set of $G$.
\end{proof}

We assume, henceforth, that \Cref{rr:twin} has been exhaustively applied to $(G,k)$.
This implies, in particular, that for every clique $C$ of $G-X$,
there are at most two vertices in $C$ with the same neighborhood in $X$.
Since the number of distinct neighborhoods in $X$ is at most $2^{|X|}$, each clique in $G-X$ has order at most $2^{|X|+1}$.
We now aim to bound the number of cliques in $G-X$.
To this end, we define a notion of \emph{equivalence classes} over the set of cliques of $G-X$.
It will easily be seen that the number of equivalence classes is at most $2^{2^{|X|+1}}$.
The number of cliques in each equivalence class will then be bounded by using \Cref{rr:identical-cliques}. 

For every clique $C$ of $G-X$, the \emph{signature} $\sign(C)$ of $C$ is the \emph{multiset} containing 
the neighborhoods in $X$ of each vertex of $C$, that is, $\sign(C) = \{N(u) \cap X: u \in C\}$.
For any two cliques $C_1,C_2$ of $G-X$, we say that $C_1$ and $C_2$ are \emph{identical},
which we denote by $C_1 \sim C_2$, 
if and only if $\sign(C_1) = \sign(C_2)$.
It is not difficult to see that $\sim$ is in fact an equivalence relation with at most $2^{2^{|X|+1}}$ equivalence classes: 
indeed, since the number of distinct neighborhoods in $X$ is at most $2^{|X|}$, 
and at most two vertices of each clique have the same neighborhood in $X$, 
the number of distinct signatures is at most $2^{2^{|X|+1}}$.
Consider now an equivalence class $\calC$ of $\sim$. 
Note that since the signature of a clique is a multiset, the number of vertices in each $C \in \calC$ is equal to $|\sign(C)|$.
For any $C_1,C_2 \in \calC$, we say that two vertices $u \in C_1$ and $v \in C_2$ are {\em clones} if $N(u) \cap X = N(v) \cap X$
(in particular, if $C_1 = C_2$ and $u \neq v$, then $u, v$ are true twins).
For any $C_1,C_2 \in \calC$ and any $u \in C_1$, we denote by $c(u,C_2)$ the set of clones of $u$ in $C_2$
(note that $|c(u,C_2)| \leq 2$).
Now observe that, for any two cliques $C_1,C_2 \in \calC$,
the number of pairs of true twins in $C_1$ and $C_2$ is the same:
we let $t(\calC)$ be the number of pairs of true twins in each clique of $\calC$.
We highlight that there are exactly $2t(\calC)$ vertices in each clique of $\calC$ that have true twins.
The following claim for clones is the analog of \Cref{obs:twins} for twins.

\begin{claim}
\label{obs:clones-make-cliques-intersect-solution}
Let $C_1$ and $C_2$ be two cliques of an equivalence class $\calC$ of $\sim$.
Let $u \in C_1$ and $v \in C_2$ be clones.
Then, for any $w \in V(G) \setminus (V(C_1) \cup V(C_2))$, $d(u,w) = d(v,w)$, and so,
for any resolving set $S$ of $G$, $S \cap (V(C_1) \cup V(C_2)) \neq \emptyset$.
\end{claim}

\begin{claimproof}
Let $w \in V(G) \setminus (V(C_1) \cup V(C_2))$. 
We show that $d(w, u) = d(w, v)$. 
Towards a contradiction and without loss of generality, assume that $d(w, u) < d(w, v)$.
Let $x$ be the first internal vertex on some shortest $(u,w)$-path $P$. 
Then, by definition, either $x \in X$ or $x \in C_1$.
Suppose first that $x \in X$.
Since $v$ is a clone of $u$, $N_X(u) = N_X(v)$, and so, $x \in N(v)$.
It follows that $d(w, v) \leq 1 + d(x, w) = 1 + d(u, w) -1 = d(u,w)$, a contradiction.
Second, suppose that $x \in C_1$ and let $y$ be the clone of $x$ in $C_2$.
Let $z$ be the vertex after $x$ on $P$.
Then, $z \not \in C_1$ since $P$ is a shortest path and $C_1$ is a clique.
It follows that $z \in X$; in particular, $z \in N_X(x) = N_X(y)$. 
But, $y, v \in C_2$, and so, $d(v, w) \leq 2 + d(z,w) = 2 + d(u,w) - 2 = d(u,w)$, a contradiction. 

Now, consider a resolving set $S$ of $G$.
Then, there exists $w \in S$ such that $d(u,w) \neq d(v,w)$. 
But, by the above paragraph, necessarily $w \in V(C_1) \cup V(C_2)$, 
and so, $S \cap (V(C_1) \cup V(C_2)) \neq \emptyset$.
\end{claimproof}

It follows from the above claim that, for any equivalence class $\calC$ of $\sim$
and any resolving set $S$, $S$ contains at least $|\calC| -1$ vertices in $V(\calC) := \bigcup_{C \in \calC} V(C)$.
We now present an upper bound on the size of $S \cap V(\calC)$ when $|\calC| \ge |X| + 2$.	

\begin{claim}
\label{obs:upper-bound-sol-interesction-eqiv-class}
For every equivalence class $\calC$ of $\sim$,
if $|\calC| \geq |X| + 2$ 
then, for any minimum resolving set $S$ of $G$, $|S \cap V(\calC)| \le |X| + |\calC| \cdot \max\{1, t(\calC)\}$.
\end{claim}

\begin{claimproof}
Let $\calC$ be an equivalence class such that $|\calC| \geq |X| + 2$.
Suppose, towards a contradiction, that there exists a minimum resolving set $S$ of $G$ 
such that $|S \cap V(\calC)| > |X| + |\calC| \cdot \max\{1,  t(\calC)\}$.
Consider the subset $S^{\circ}$ of $V(G)$ obtained from $S$ as follows:
\begin{enumerate}
\item add every vertex of $S$ to $S^\circ$;
\item delete all vertices of $S \cap V(\calC)$ from $S^\circ$, and add all the vertices of $X$ to $S^{\circ}$;
\item if $t(\calC) = 0$ then, for every clique $C \in \calC$, add an arbitrary vertex of $C$ to $S^{\circ}$, and 
otherwise, for every pair of twin vertices in $C \in \calC$, add one of them to $S^{\circ}$. 
\end{enumerate}
Now, note that the second step removes at least  $|X| + |\calC| \cdot \max\{1,  t(\calC)\} + 1$ vertices from $S^\circ$
and adds exactly $|X|$ vertices to $S^\circ$.
Since the third step adds $|\calC| \cdot \max\{1,  t(\calC)\}$ vertices to $S^\circ$, it follows that $|S^{\circ}| < |S|$. 
Let us next prove that $S^{\circ}$ is a resolving set of $G$, which if true would contradict the minimality of $S$, 
and thus, conclude the proof.

Consider any two distinct vertices $u,v \in V(G)$.
Since $X \subseteq S^\circ$, if $\{u,v\} \cap X \neq \emptyset$, then any vertex in $\{u,v\} \cap X$ resolves $u$ and $v$.
Thus, we may assume that $u,v \in V(G) \setminus X$.
Suppose first that one of $u$ and $v$ belongs to $V(\calC)$, say $u \in V(\calC)$ without loss of generality.
Let $C \in \calC$ be the clique such that $u \in V(C)$.
If $v \notin V(C)$, then the vertex in $S^\circ \cap V(C)$ surely resolves $u$ and $v$.
Thus, suppose that $v \in V(C)$ as well.
If $N(u) \cap X \neq N(v) \cap X$, then any vertex in $(N(u) \setminus N(v)) \cap X$ or $(N(v) \setminus N(u)) \cap X$ resolves $u$ and $v$.
Otherwise, $u$ and $v$ are true twins, and so, by construction, one of them belongs to $S^\circ$ which resolves the pair.
Finally, suppose that $u,v \notin X \cup V(\calC)$.
Since $S$ is a resolving set of $G$, there exists $w \in S$ such that $d(w,u) \neq d(w,v)$, 
say $d(w,u) < d(w,v)$ without loss of generality.
If $w$ belongs to $S^\circ$ as well, then we are done.
Thus, suppose that $w \notin S^\circ$.
Then, $w \in S \cap V(\calC)$ by construction.
Let $P$ be a shortest path from $u$ to $w$.
Since $u \notin X \cup V(\calC)$, necessarily $V(P) \cap X \neq \emptyset$:
let $x \in X$ be the closest internal vertex of $P$ to $w$.
Then, $\dist(w, u) = \dist(w, x) + \dist(x, u)$, and since $d(w,u) < d(w,v) \leq d(w,x) + d(x,v)$,
we conclude that $d(x,u) < d(x,v)$.
Since $X \subseteq S^\circ$, it follows that there is a vertex in $S^\circ$ that resolves $u$ and $v$.
Therefore, $S^\circ$ is a resolving set of $G$, which concludes the proof.
\end{claimproof}

Let $\calC$ be an equivalence class of $\sim$, and let $S$ be a resolving set of $G$.
For every $i \geq 0$, we denote by $\calC^S_{=i}$ ($\calC^S_{\geq i}$, respectively) the set of cliques $C \in \calC$ 
such that $|S \cap V(C)| = i$ ($|S \cap V(C)| \geq i$, respectively).

\begin{claim}
\label{obs:bounding-partition}
Let $\calC$ be an equivalence class of $\sim$ such that $|\calC| \geq |X| + 2$.
Then, for any minimum resolving set $S$ of $G$, the following hold:
\begin{itemize}
\item[(i)] if $t(\calC) = 0$, then $|\calC^S_{=0}| \le 1$ and $|\calC^S_{\ge 2}| \le |X| + 1$;
\item[(ii)] if $t(\calC) \neq 0$, then $|\calC^S_{\ge t(\calC) + 1}| \le |X| + 1$.
\end{itemize}
\end{claim}

\begin{claimproof}
Suppose first that $t(\calC) = 0$. 
Then, by \Cref{obs:clones-make-cliques-intersect-solution}, $|\calC^S_{=0}| \le 1$.
Since $|\calC^S_{=0}| + |\calC^S_{=1}| + |\calC^S_{\geq 2}| = |\calC|$, it follows that 
$|\calC^S_{=1}| \ge |\calC| - 1 - |\calC^S_{\ge 2}|$.
Now, by definition, $|S \cap V(\calC)| \ge |\calC^S_{=1}| + 2 \cdot |\calC^S_{\ge 2}|$, and so, by the above,
$$|S \cap V(\calC)| \ge |\calC| - 1 - |\calC^S_{\ge 2}| + 2 \cdot |\calC^S_{\ge 2}| = |\calC| - 1 + |\calC^S_{\ge 2}|.$$
Thus, if $|\calC^S_{\ge 2}| \ge |X| + 2$, then $|S \cap V(\calC)| \ge |\calC| - 1 + |X| + 2$, 
a contradiction to \Cref{obs:upper-bound-sol-interesction-eqiv-class}.

Second, suppose that $t(\calC) \neq 0$.
Then, by \Cref{obs:twins}, any resolving set of $G$ contains at least $t(\calC)$ vertices from each clique in $\calC$,
which implies, in particular, that $(\calC^S_{=t(\calC)}, \calC^S_{\ge t(\calC) + 1})$ is a partition of $\calC$. 
Now, by definition, 
\begin{equation*}
\begin{split}
|S \cap V(\calC)| &\ge t(\calC)  \cdot |C^S_{=t(\calC)}| + (t(\calC) + 1) \cdot |\calC^S_{\ge t(\calC) + 1}| 
=\\ 
&= t(\calC) \cdot (|\calC^S_{=t(\calC)}| + |\calC^S_{\ge t(\calC) + 1}|) + |\calC^S_{\ge t(\calC) + 1}| 
= \\ 
&= t(\calC)\cdot |\calC| +  |\calC^S_{\ge t(\calC) + 1}|.
\end{split}
\end{equation*}
Thus, if $|\calC^S_{\ge t(\calC) + 1}| \ge |X| + 2$, then $|S \cap V(\calC)| \ge t(\calC) \cdot |\calC| + |X| + 2$,
a contradiction to \Cref{obs:upper-bound-sol-interesction-eqiv-class}.
\end{claimproof}

\Cref{obs:clones-make-cliques-intersect-solution} and \Cref{obs:bounding-partition} together imply that if some equivalence class $\calC$ of $\sim$ contains at least $|X| + 3$ cliques,
then, for any minimum resolving set $S$ of $G$, 
if $t(\calC) = 0$, then $\calC^S_{=1} \neq \emptyset$, and otherwise, $\calC^S_{= t(\calC)} \neq \emptyset$.
The following reduction rule is based on this fact.

\begin{reduction rule}
\label{rr:identical-cliques}
If there exists an equivalence class $\calC$ of $\sim$ such that $|\calC| \ge 2^{|X|+2} +  |X| + 2$,  
then remove a clique $C \in \calC$ from $G$ and reduce $k$ by $\max\{1,t(\calC)\}$.
\end{reduction rule}

We next prove that \Cref{rr:identical-cliques} is safe.
To this end, we first prove the following.

\begin{claim}\label{clm:clones}
Let $C_1$ and $C_2$ be two identical cliques. 
Then, for every $u_1 \in V(C_1)$ and $v_2 \in V(C_2)$, $d(u_1,v_2) = d(u_2,v_1)$, 
where $u_2 \in c(u_1,C_2)$ and $v_1 \in c(v_2,C_1)$. 
\end{claim}

\begin{claimproof}
Consider $u_1 \in V(C_1)$ and $v_2 \in V(C_2)$ and let $P = w_1\ldots w_p$, where $u_1 = w_1$ and $v_2 = w_p$,
be a shortest path from $u_1$ to $v_2$.
Let $u_2 \in c(u_1,C_2)$ and $v_1 \in c(v_2,C_1)$.
Suppose first that $w_2$ belongs to $X$.
If $w_{p-1}$ belongs to $X$ as well,
then $u_2P[w_2,w_{p-1}]v_1$ is a path from $u_2$ to $v_1$ 
since $N_X(u_1) = N_X(u_2)$ and $N_X(v_1) = N_X(v_2)$.
Suppose next that $w_{p-1}$ belongs to $C_2$ and let $t \in c(w_{p-1},C_1)$.
Since $P$ is a shortest path and $C_2$ is a clique, $w_{p-2}$ must belong to $X$.
Then, the path $u_2P[w_2,w_{p-2}]tv_1$ is a path from $u_2$ to $v_1$ since $N_X(w_{p-1}) = N_X(t)$.
Second, suppose that $w_2$ belongs to $C_1$ and let $t \in c(w_2,C_2)$.
By symmetry, we may assume that $w_{p-1}$ belongs to $C_2$ (we fall back into the previous case otherwise), 
and let $r \in c(w_{p-1},C_1)$.
Since $P$ is a shortest path and $C_1,C_2$ are cliques, it must be that $w_3,w_{p-2} \in X$, and thus, 
the path $u_2tP[w_3,w_{p-2}]rv_1$ is a path from $u_2$ to $v_1$ 
since $N_X(w_2) = N_X(t)$ and $N_X(w_{p-1}) = N_X(r)$.
In all cases, we obtain that $d(u_2,v_1) \leq d(u_1,v_2)$,  
and by symmetry, we conclude that in fact equality holds.
\end{claimproof}

\begin{lemma}
\label{claim:rr-idnetical-cliques-safe}
\Cref{rr:identical-cliques} is safe.
\end{lemma}

\begin{proof}
Assume that there exists an equivalence class $\calC$ of $\sim$ such that $|\calC| \ge 2^{|X|+2} +  |X| + 2$ 
and consider a clique $C_1 \in \calC$. Then the following hold.

\begin{claim}
\label{clm:condition_res}
There exists a minimum resolving set $S$ of $G$ such that 
if $t(\calC) = 0$, then $C_1 \in \calC^S_{=1}$, and otherwise, $C_1 \in \calC^S_{=t(\calC)}$.
\end{claim}

\begin{claimproof}
Consider a minimum resolving set $S$ of $G$.
Assume that $C_1 \notin \calC^S_{=\max\{1,t(\calC)\}}$ (we are done otherwise).
Since $|\calC| \ge 2^{|X|+2} +  |X| + 2$, then $\calC^S_{=\max\{1,t(\calC)\}} \neq \emptyset$ by \Cref{obs:bounding-partition}, and so, let $C_2 \in \calC^S_{=\max\{1,t(\calC)\}}$.
We now construct a resolving set $S^\circ$ from $S$ as follows.
First, we add every vertex in $S \setminus (V(C_1) \cup V(C_2))$ to $S^\circ$.
Then, for every $x \in V(C_1) \cap S$, we add a clone of $x$ in $C_2$ to $S^\circ$.
Similarly, for every $x \in V(C_2) \cap S$, we add a clone of $x$ in $C_1$ to $S^\circ$.
Let us show that $S^\circ$ is indeed a resolving set of $G$.
Suppose to the contrary that there exist $u,v\in V(G)$ such that no vertex in $S^\circ$ resolves $u$ and $v$.
Assume first that $u,v \notin V(C_1) \cup V(C_2)$.
Since $S$ is a resolving set of $G$, there exists $w \in S \setminus S^\circ$ such that $d(w,u) \neq d(w,v)$. 
In particular, either $w \in V(C_1)$ or $w \in V(C_2)$.
Now, if $w \in V(C_1)$, then a clone $w_2 \in V(C_2)$ of $w$ belongs to $S^\circ$ by construction.
But, by \Cref{obs:clones-make-cliques-intersect-solution}, $d(w_2,u) = d(w,u) \neq d(w,v) = d(w_2,v)$, 
a contradiction to our assumption.
By symmetry, we conclude similarly if $w \in V(C_2)$.
Thus, at least one of $u$ and $v$ belongs to $V(C_1) \cup V(C_2)$.

Suppose first that exactly one of $u$ and $v$ belongs to $V(C_1) \cup V(C_2)$, 
say $u \in V(C_1) \cup V(C_2)$ without loss of generality .
Let us suppose that $u \in V(C_1)$ (the case where $u \in V(C_2)$ is handled symmetrically).
Since $S$ is a resolving set of $G$, there exists $w \in S$ such that $d(w,u_2) \neq d(w,v)$, 
where $u_2 \in c(u,C_2)$.
If $w \in S^\circ$, then $w \notin V(C_1) \cup V(C_2)$ by construction, and so,
by \Cref{obs:clones-make-cliques-intersect-solution}, $d(w,u) = d(w,u_2) \neq d(w,v)$,
a contradiction to our assumption.
Suppose therefore that $w \notin S^\circ$.
If $w \in V(C_2)$, then a clone $w_1 \in c(w,C_1)$ belongs to $S^\circ$ by construction. 
But, by \Cref{obs:clones-make-cliques-intersect-solution}, $d(w_1,v) = d(w,v)  \neq d(w,u_2) = d(w_1,u)$,
a contradiction to our assumption.
Similarly, if $w \in V(C_1)$, then a clone $w_2 \in c(w,C_2)$ belongs to $S^\circ$ by construction. 
But, by \Cref{obs:clones-make-cliques-intersect-solution} and \Cref{clm:clones}, $d(w_2,v) = d(w,v)  \neq d(w,u_2) = d(w_2,u)$,
a contradiction to our assumption.

Second, suppose that both $u$ and $v$ belong to $V(C_1) \cup V(C_2)$.
Then, either (1) $u$ and $v$ belong to the same clique or (2) $u$ and $v$ belong to different cliques.

Assume first that (1) holds, say $u,v \in V(C_1)$ without loss of generality.  
Since $S$ is a resolving set of $G$, there exists $w \in S$ such that $d(w,u_2) \neq d(w,v_2)$, 
where $u_2,v_2 \in V(C_2)$ are clones of $u,v$, respectively.
If $w \in S^\circ$, then $w \notin V(C_1) \cup V(C_2)$ by construction, and so,
by \Cref{obs:clones-make-cliques-intersect-solution}, $d(w,u) = d(w,u_2) \neq d(w,v_2) = d(w,v)$,
a contradiction to our assumption.
Suppose therefore that $w \notin S^\circ$.
Since $C_2$ is a clique, it must then be that $w \in V(C_1)$. 
But, a clone $w_2 \in c(w,C_2)$ belongs to $S^\circ$ by construction, 
and $d(w_2,u) = d(w,u_2) \neq d(w,v_2) = d(w_2,v)$ by \Cref{clm:clones},
a contradiction to our assumption.

Assume finally that (2) holds, say $u \in V(C_1)$ and $v \in V(C_2)$ without loss of generality
Since $S$ is a resolving set, there exists $w \in S$ such that $d(w,u_2) \neq d(w,v_1)$,
where $u_2 \in V(C_2)$ and $v_1 \in V(C_1)$ are clones of $u$ and $v$, respectively.
If $w \in S^\circ$, then $w \notin V(C_1) \cup V(C_2)$ by construction, and so,
by \Cref{obs:clones-make-cliques-intersect-solution}, $d(w,u) = d(w,u_2) \neq d(w,v_1) = d(w,v)$,
a contradiction to our assumption.
Suppose therefore that $w \notin S^\circ$.
If $w \in V(C_1)$, then a clone $w_2 \in c(w,C_2)$ belongs to $S^\circ$ by construction. 
But, by \Cref{clm:clones}, $d(w_2,u) = d(w,u_2) \neq d(w,v_1) = d(w_2,v)$, a contradiction.
By symmetry, we conclude similarly if $w \in V(C_2)$.  
Therefore, $S^\circ$ is a resolving set of $G$, and since $C_1 \in \calC^{S^\circ}_{=\max\{1,t(\calC)\}}$ by construction,
our claim follows.
\end{claimproof}
 
In the following, for simplicity, we let $C = C_1$. Let us next show that $G$ has a resolving set of size at most $k$ if and only if 
$G- V(C)$ has a resolving set of size at most $k - \max\{1,t(\calC)\}$. We prove the two directions in two separate claims.

\begin{claim}
\label{clm:GtoG-VC}
If $G$ has a resolving set of size at most $k$, then $G- V(C)$ has a resolving set of size at most $k - \max\{1,t(\calC)\}$.
\end{claim}

\begin{claimproof}
Let $S$ be a resolving set of $G$ of size at most $k$.
By \Cref{clm:condition_res}, we may assume that if $t(\calC) = 0$, then $C \in \calC^S_{=1}$, and otherwise $C \in \calC^S_{=t(\calC)}$.
Now, suppose that $R = S\setminus V(C)$ is not a resolving set of $G- V(C)$ (we are done otherwise).
Then, there exist $u,v \in V(G) \setminus V(C)$ such that no vertex in $R$ resolves $u$ and $v$.
Since $S$ is a resolving set of $G$, there exists $w \in S \setminus R$ such that $d(w,u) \neq d(w,v)$, and  
in particular, $w \in V(C)$.
It follows that no clone $w_B$ of $w$ in a clique $B \in \calC \setminus \{C\}$ belongs to $S$,
as otherwise, by \Cref{obs:clones-make-cliques-intersect-solution}, $d(w_B,u) = d(w,u) \neq d(w,v) = d(w_B,v)$,
a contradiction to our assumption.
Since $S$ contains at least one vertex from each pair of true twins, this implies in particular that $w$ has no true twin.
It follows that $t(\calC) = 0$. 
Indeed, if $t(\calC) \neq 0$, then $|S \cap V(C)| \geq t(\calC) + 1$ 
since $S \cap V(C)$ contains $w$ and at least one vertex from each pair of true twins, a contradiction to the choice of $C$.
Now, $|\calC| \geq 2^{|X|+2} +  |X| + 2$, and so, by \Cref{obs:bounding-partition}, $|\calC^S_{=1}| \geq 2^{|X|+2}$.
It follows that there exist four distinct cliques $C_1,C_2,C_3,C_4 \in \calC^S_{=1}$ 
such that the vertex in $S \cap V(C_4)$ is a clone of the vertex in $S \cap V(C_i)$ for every $i \in [3]$.
Without loss of generality, we may assume that $C \neq C_1,C_2,C_3$.
Let $z \in S \cap V(C_1)$ and let $w_1 \in V(C_1)$ be a clone of $w$ (recall that $w \in S \cap V(C)$).
We claim that the set $R^*$ obtained from $R$ by replacing $z$ with $w_1$ is a resolving set of $G - V(C)$.
Indeed, suppose to the contrary that there exist $x,y \in V(G) \setminus V(C)$ such that no vertex in $R^*$ resolves $x$ and $y$.
Then, either (1) $w$ resolves $x$ and $y$ or (2) $z$ resolves $x$ and $y$.

Assume first that (1) holds.
If $x,y \notin V(C_1)$, then by \Cref{obs:clones-make-cliques-intersect-solution}, 
$d(w_1,x) = d(w,x) \neq d(w,y) = d(w_1,y)$,
a contradiction to our assumption. 
It follows that at least one of $x$ and $y$ belongs to $C_1$.
Suppose first that exactly one of $x$ and $y$ belongs to $C_1$, say $x\in V(C_1)$ without loss of generality
Since $S$ is a resolving set of $G$, there exists $t \in S$ such that $d(t,x^\circ) \neq d(t,y)$, 
where $x^\circ \in c(x,C)$.
If $t = w$, then by \Cref{obs:clones-make-cliques-intersect-solution}, $d(w_1,y) = d(w,y) \neq d(w,x^\circ) = d(w_1,x)$,
a contradiction to our assumption.
Now, suppose that $t = z$ and let $z_2 \in c(z,C_2)$ (recall that $z_2 \in S \cap R^*$).
If $y \in V(C_2)$, then $d(z_2,y) \leq 1 \neq d(z_2,x)$, and 
if $y \notin V(C_2)$, then by \Cref{obs:clones-make-cliques-intersect-solution}, $d(z_2,y) = d(z,y) \neq d(z,x^\circ) = d(z_2,x)$,
a contradiction in both cases to our assumption.
Thus, $t \notin V(C) \cup V(C_1)$. 
But then, by \Cref{obs:clones-make-cliques-intersect-solution}, $d(t,x) = d(t,x^\circ) \neq d(t,y)$, 
a contradiction to our assumption.
Suppose therefore that both $x$ and $y$ belong to $C_1$.
Since $S$ is a resolving set of $G$, there exists $t \in S$ such that $d(t,x^\circ) \neq d(t,y^\circ)$,
where $x^\circ,y^\circ \in V(C)$ are the clones of $x,y$ respectively.
If $t = w$, then either $x^\circ = w$ or $y^\circ = w$, say the latter holds without loss of generality 
But then $d(w_1,y) = 0 < d(w_1,x)$, a contradiction to our assumption.
If $t = z$, then by \Cref{obs:clones-make-cliques-intersect-solution}, $d(z_2,x) = d(z,x^\circ) \neq d(z,y^\circ) = d(z_2,y)$, 
where $z_2 \in V(C_2)$ is the clone of $z$,
a contradiction to our assumption.
Thus, $t \neq z,w$, and so,
by \Cref{obs:clones-make-cliques-intersect-solution}, $d(t,x) = d(t,x^\circ) \neq d(t,y^\circ) = d(t,y)$,
a contradiction to our assumption.

Assume now that (2) holds and, for all $i \in \{2,3\}$, let $z_i \in V(C_i)$ be the clone of $z$.
Then, at least one of $x$ and $y$ belongs to $V(C_1) \cup V(C_2)$, 
as otherwise, by \Cref{obs:clones-make-cliques-intersect-solution}, $d(z_2,x) = d(z,x) \neq d(z,y) = d(z_2,y)$.
The same argument shows that at least one of $x$ and $y$ belongs to $V(C_1) \cup V(C_3)$.
It follows that either one of $x$ and $y$ belongs to $C_2$, while the other belongs to $C_3$, 
in which case $z_2$ resolves $x$ and $y$, 
or at least one of $x$ and $y$ belongs to $C_1$.
Suppose first that $x,y \in V(C_1)$.
Then, either $x = z$ or $y = z$, say the latter holds without loss of generality 
In particular, $z \neq w_1$.
Since $S$ is a resolving set of $G$, there exists $t \in S$ such that $d(t,x^\circ) \neq d(t,y^\circ)$, 
where $x^\circ,y^\circ \in V(C)$ are the clones of $x,y$, respectively.
Since $z \neq w_1$, $t \neq w$ by construction, and 
in particular, $t \notin V(C)$.
Now, if $t = z$, then by \Cref{obs:clones-make-cliques-intersect-solution},
$d(z_2,x) = d(z,x^\circ) \neq d(z,y^\circ) = d(z_2,y)$, 
and if $t \neq z$, then $t \in S \setminus (V(C_1) \cup V(C)) \subseteq R^*$, and so, $t$ resolves $x$ and $y$
since by \Cref{obs:clones-make-cliques-intersect-solution}, $d(t,x) = d(t,x^\circ) \neq d(t,y^\circ) = d(t,y)$,
a contradiction in both cases to our assumption. 
Thus, it must be that exactly one of $x$ and $y$ belongs to $C_1$, say $x \in V(C_1)$ without loss of generality
Then, $y \notin V(C_2)$, as otherwise, $z_2$ resolves $x$ and $y$,
and similarly, $y \notin V(C_3)$.
Now, since $S$ is a resolving set of $G$, there exists $t \in S$ such that $d(t,x^\circ) \neq d(t,y)$,  
where $x^\circ \in V(C)$ is the clone of $x$.
If $t = w$, then by \Cref{obs:clones-make-cliques-intersect-solution}, $d(w_1,x) = d(w,x^\circ) \neq d(w,y) = d(w_1,y)$,
a contradiction to our assumption.
It $t = z$, then by \Cref{obs:clones-make-cliques-intersect-solution}, $d(z_2,x) = d(z,x^\circ) \neq d(z,y) = d(z_2,y)$,
a contradiction to our assumption.
Otherwise, $t \in S \setminus (V(C_1) \cup V(C)) \subseteq R^*$, and so, $t$ resolves $x$ and $y$
since by \Cref{obs:clones-make-cliques-intersect-solution}, $d(t,x) = d(t,x^\circ) \neq d(t,y)$,
a contradiction to our assumption.
Thus, $R^*$ is indeed a resolving set of $G- V(C)$ of size at most $k-1$.
\end{claimproof}

\begin{claim}
\label{clm:G-VCtoG}
If $G-V(C)$ has a resolving set of size at most $k - \max\{1,t(\calC)\}$, then $G$ has a resolving set of size at most $k$.
\end{claim}

\begin{claimproof}
Let $S$ be a resolving set of $G-V(C)$ of size at most $k - \max\{1,t(\calC)\}$.
Consider the set $R$ constructed from $S$ as follows.
First, add every vertex of $S$ to $R$.
Then, letting $B^\circ$ be a clique of $\calC^S_{=\max\{1,t(\calC)\}}$
(note that since $|\calC| \geq 2^{|X|+2} +  |X| + 2$, it follows from \Cref{obs:bounding-partition} that 
$\calC^S_{=\max\{1,t(\calC\}} \neq \emptyset$),
for every $x^\circ \in S \cap V(B^\circ)$, we add a clone $x \in V(C)$ of $x^\circ$ to $R$.
We contend that $R$ is a resolving set of $G$.
Indeed, consider $u,v \in V(G)$.
We may assume that $R \cap \{u,v\} = \emptyset$, as otherwise it is clear that a vertex in $R$ resolves $u$ and $v$.
Now, if $\{u,v\} \cap V(C) = \emptyset$, then there exists a vertex $w \in S \subseteq R$ that resolves $u$ and $v$ 
since $S$ is a resolving set of $G$.
Assume next that $|\{u,v\} \cap V(C)| =1$, say $u \in V(C)$ without loss of generality
Suppose that no vertex in $R \cap V(C)$ resolves $u$ and $v$. 
Note that in this case, $v \in X$ and, for every $x \in R \cap V(C)$, $d(x,u) = 1 = d(x,v)$.
Let $u^\circ \in V(B^\circ)$ be the clone of $u$, where $B^\circ \in \calC^S_{=\max\{1,t(\calC)\}}$ is the clique considered in the construction of $R$.
Since $S$ is a resolving set of $G-V(C)$, there exists $t \in S$ such that $d(t,u^\circ) \neq d(t,v)$.
If $t \notin V(B^\circ)$, then by \Cref{obs:clones-make-cliques-intersect-solution}, $d(t,u) = d(t,u^\circ) \neq d(t,v)$,
that is, $t \in R$ resolves $u$ and $v$.
Suppose therefore that $t \in V(B^\circ)$. 
Then, it must be that $t = u^\circ$. 
Indeed, if $t \neq u^\circ$, then $d(t,u^\circ) =1$, 
but since for all $x \in R \cap V(C)$, $d(x,u) = 1 = d(x,v)$, it follows by construction that $d(t,v) =1$, a contradiction.
Thus, $t$ must have a true twin $t'$. 
Indeed, if $t$ has no true twin, then by construction, $u \in R$, a contradiction to our assumption.
Now, since $B^\circ \in \calC^S_{=\max\{1,t(\calC)\}}$, $t' \notin S$. 
But, since, for every $x \in R \cap V(C)$, $d(x,v) = 1$, it follows by construction that, for every $y\in S \cap V(B^\circ)$,
$d(y,v) = 1 = d(y,t')$.  
Thus, there must exist $w \in S \setminus V(B^\circ)$ such that $d(w,t') \neq d(w,v)$. 
But then, by \Cref{obs:clones-make-cliques-intersect-solution}, $d(w,u) = d(w,t') \neq d(w,v)$,
that is, $w \in R$ resolves $u$ and $v$.
Suppose finally that both $u$ and $v$ belong to $C$.
We claim that there exists a clique $B \in \calC \setminus \{C\}$ such that 
a clone $u_B$ of $u$ in $B$ and a clone $v_B$ of $v$ in $B$ are resolved by a vertex in $S \setminus V(B)$.
Indeed, towards a contradiction suppose that the statement does not hold.
Then, for every clique $B \in \calC \setminus \{C\}$, every $u_B \in c(u,B)$, and $v_B \in c(v,B)$,
$u_B$ and $v_B$ are resolved by a vertex $w$ in $V(B)$. 
Also, since $d(x,u_B) = d(x,v_B)$ for any $x \in V(B) \setminus \{u_B,v_B\}$,
we conclude that in fact $w \in \{u_B,v_B\}$.
It follows that $t(\calC) \neq 0$. 
Indeed, if $t(\calC) = 0$, 
then, in particular, $S \cap V(B^\circ) \subseteq \{u_{B^\circ},v_{B^\circ}\}$,  
where $B^\circ \in \calC^S_{=\max\{1,t(\calC)\}}$ is the clique considered in the construction of $R$. 
But, $u,v \notin R$ by assumption, a contradiction to the choice of $R$.
Now, if neither $u$ nor $v$ have true twins, then since $S$ contains at least one vertex from each pair of twins,
\begin{equation*}
\begin{split}
|S \cap V(\calC \setminus \{C\})| &\geq \bigg|S \cap \bigcup_{B \in \calC \setminus \{C\}} c(u,B) \cup c(v,B)\bigg| + (|\calC| - 1) \cdot t(\calC)\\
&\geq |\calC \setminus \{C\}| + (|\calC| - 1) \cdot t(\calC) \\
& \geq  2^{|X|+2} +  |X| + 1 + (|\calC| - 1) \cdot t(\calC),
\end{split}
\end{equation*}
a contradiction to \Cref{obs:upper-bound-sol-interesction-eqiv-class}.
Similarly, if exactly one of $u$ and $v$ has a true twin, then since $S$ contains at least one vertex from every other pair of twins,
\begin{equation*}
\begin{split}
|S \cap V(\calC \setminus \{C\})| &\geq \bigg|S \cap \bigcup_{B \in \calC \setminus \{C\}} c(u,B) \cup c(v,B)\bigg| + (|\calC| - 1) \cdot (t(\calC) -1)\\
&\geq 2|\calC \setminus \{C\}| + (|\calC| - 1) \cdot (t(\calC)-1)\\
& \geq 2(2^{|X|+2} +  |X| + 1) + (|\calC| - 1) \cdot t(\calC-1),
\end{split}
\end{equation*}
a contradiction to \Cref{obs:upper-bound-sol-interesction-eqiv-class}.
Finally, if both $u$ and $v$ have true twins, then since $S$ contains at least one vertex from every other pair of twins,
\begin{equation*}
\begin{split}
|S \cap V(\calC \setminus \{C\})| &\geq \bigg|S \cap \bigcup_{B \in \calC \setminus \{C\}} c(u,B) \cup c(v,B)\bigg| + (|\calC| - 1) \cdot (t(\calC) -2)\\
&\geq 3|\calC \setminus \{C\}| + (|\calC| - 1) \cdot (t(\calC)-2)\\ 
&\geq 3(2^{|X|+2} +  |X| + 1) + (|\calC| - 1) \cdot t(\calC-2),
\end{split}
\end{equation*}
a contradiction to \Cref{obs:upper-bound-sol-interesction-eqiv-class}.
Thus, as claimed, there exists a clique $B \in \calC \setminus \{C\}$ such that 
a clone $u_B$ of $u$ in $B$ and a clone $v_B$ of $v$ in $B$ are resolved by a vertex $w\in S \setminus V(B)$. 
Also, since $d(w,u) = d(w,u_B) \neq d(w,v_B) = d(w,v)$ by \Cref{obs:clones-make-cliques-intersect-solution}, 
we conclude that $w$ resolves $u$ and $v$.
Since $R$ has size at most $k$, the claim follows.
\end{claimproof}

The safeness of \Cref{rr:identical-cliques} now follows from \Cref{clm:GtoG-VC} and \Cref{clm:G-VCtoG}.
\end{proof}

Now, observe that once \Cref{rr:identical-cliques} has been exhaustively applied to $(G,k)$, 
each equivalence class of $\sim$ contains at most $2^{|X|+2} +  |X| + 1$ cliques.
Since there are at most $2^{2^{|X|+1}}$ equivalence classes and each clique of $G-X$ has size at most $2^{|X|+1}$,
we conclude that $G$ contains at most $2^{2^{|X|+1}} \cdot (2^{|X|+2} +  |X| + 1) \cdot 2^{|X|+1} + |X|$ vertices.
\end{proof}

We now briefly explain how to adapt the proof for the distance to co-cluster.
First, recall that the {\it distance to co-cluster} of a graph $G$ is 
the minimum number of vertices of $G$ that need to be deleted 
so that the resulting graph is a {\it co-cluster graph}, {\it i.e.}, the complement of a cluster graph. 

Now, let $G$ be a graph, and let $X \subseteq V(G)$ be such that $G-X$ is a co-cluster graph.
Then, $G-X$ consists of a disjoint union of independent sets 
such that for any two distinct such independent sets $I_1,I_2$, 
every vertex of $I_1$ is adjacent to every vertex of $I_2$;
in particular, any two vertices in a same independent set are false twins in $G-X$.
Similarly to the distance to cluster proof, we first exhaustively apply \Cref{rr:twin} to \emph{false} twins in $G$.
Then, for every independent set $I$ of $G-X$, 
there are at most two vertices of $I$ with the same neighborhood in $X$,
which implies that each independent set of $G-X$ has order at most $2 \cdot 2^{|X|}$
(indeed, the number of distinct neighborhoods in $X$ is at most $2^{|X|}$).
To bound the number of independent sets in $G-X$,
we define an equivalence relation over these sets in a similar fashion to the distance to cluster.
More specifically, for every independent set $I$ of $G-X$, 
we let the \emph{signature} $\sign(I)$ of $I$ be 
the \emph{multiset} containing the neighborhoods in $X$ of each vertex of $I$,
that is, $\sign(I) = \{N(u) \cap X:u \in I\}$.
Then, two independent sets $I_1,I_2$ of $G-X$ are said to be \emph{identical},
which we denote by $I_1 \sim I_2$,  if $\sign(I_1) = \sign(I_2)$.
As for the distance to cluster, the relation $\sim$ is in fact an equivalence relation, 
with at most $2^{2^{|X|+1}}$ equivalence classes.
Given an equivalence class $\calC$ of $\sim$, 
we may similarly define $t(\calC)$ as the number of pairs of \emph{false} twins in each independent set of $\calC$.
We then bound the number of independent sets in each equivalence of $\sim$ using the following reduction rule 
(it is the analog of \Cref{rr:identical-cliques}).

\begin{reduction rule}
\label{rr:identical-independent}
If there exists an equivalence class $\calC$ of $\sim$ such that $|\calC| \geq 2^{|X|+2} + |X| +2$, then remove an independent set $I \in \calC$ from $G$ and reduce $k$ by $\max\{1,t(\calC)\}$. 
\end{reduction rule}

The proof of the safeness of \Cref{rr:identical-independent} is then analogous to that of \Cref{rr:identical-cliques};
and we can similarly argue that after \Cref{rr:identical-independent} has been exhaustively applied,
$G$ contains at most $2^{2^{|X|+1}} \cdot (2^{|X|+2}+|X|+1) \cdot 2^{|X|+1} + |X|$ vertices.
\section{Conclusion}
As the \textsc{Metric Dimension} problem is \W[2]-hard when parameterized by the solution size~\cite{HartungN13}, the next natural step is to understand its parameterized complexity under structural parameterizations.
We continued this line of research, following the steps of~\cite{BelmonteFGR17,E15,GHK22}, and more recently~\cite{BP21,LM22}.
Our most technical result is a proof that the \textsc{Metric Dimension} problem is \W[1]-hard when parameterized by the combined parameter feedback vertex set number plus pathwidth of the graph.
We thereby improved the result by Bonnet and Purohit~\cite{BP21} that states the problem is \W[1]-hard when parameterized by the pathwidth, and answered an open question in~\cite{HartungN13}.
It is easy to see that the problem admits an \FPT\ algorithm when parameterized by the larger parameter, the vertex cover number of the graph.

Although this work advances the understanding of structural parameterizations of \textsc{Metric Dimension}, it falls short of completing the picture (see Figure~\ref{fig:result-overview}).
We find it hard to extend the positive results to the parameters like distance to disjoint paths, feedback edge set, and bandwidth.
It would be interesting to find \FPT\ algorithms or prove that such algorithms are highly unlikely to exist for these parameters.
The \FPT\ algorithm parameterized by treedepth in~\cite{GHK22} relies on a meta-result.
Is it possible to get an \FPT\ algorithm whose running time is a single or double exponent in the treedepth?
It would also be interesting to investigate the parameterized complexity of the problem when the parameter is the distance to cograph.
Recall that the problem is polynomial-time solvable in cographs~\cite{ELW15}.

Bonnet and Purohit~\cite{BP21} conjectured that the problem is \W[1]-hard even when parameterized by the treewidth plus the solution size.
Towards resolving this conjecture, an interesting question would be to investigate whether the problem admits an \FPT\ algorithm when parameterized by the feedback vertex set number plus the solution size.
Note that even an \XP\ algorithm parameterized by the feedback vertex set number is not apparent.

\subsection*{Acknowledgement}
In the extended abstract of this paper~\cite{MFCSversion}, we proved that the problem does not admit a polynomial kernel when parameterized by the vertex cover number or the distance to clique unless $\NP \subseteq \coNP/poly$.
The authors would like to thank Florent Foucaud for pointing us to Gutin et al.~\cite{DBLP:journals/tcs/GutinRRW20}, which contains a slightly stronger result.

\bibliographystyle{unsrt}
\bibliography{metric-dim-bib}

\end{document}